\documentclass{article}
\usepackage{fancyhdr}
\usepackage[T1]{fontenc}
\usepackage{amsmath}
\usepackage{amssymb}
\usepackage{amsthm}
\usepackage{graphicx}
\usepackage{color}
\usepackage{colortbl}
\usepackage[dvips]{epsfig}
\usepackage{xspace}
\usepackage{algorithm}
\usepackage{algorithmicx}
\usepackage{algpseudocode}
\usepackage{float}
\usepackage{lineno}
\newfloat{algorithm}{t}{lop}
\usepackage{floatflt}
\definecolor{grey}{gray}{0.7}
\usepackage[margin=1in]{geometry}
\usepackage{enumitem}
\usepackage{url}
\usepackage{hyperref}

\newtheorem{example}{Example}[section]
\newtheorem{theorem}[example]{Theorem}
\newtheorem{notation}[example]{Notation}
\newtheorem{observation}[example]{Observation}
\newtheorem{definition}[example]{Definition}
\newtheorem{proposition}[example]{Proposition}

\newtheorem{corollary}[example]{Corollary}

\newtheorem{lemma}[example]{Lemma}

\newcommand{\RR}{\mathbb{R}}

\newcommand{\bigO}{\mathcal{O}}

\newcommand{\ie}{i.\,e.,\xspace}
\newcommand{\eg}{e.\,g.,\xspace}
\newcommand{\etal}{et al.\xspace}

\newcommand{\wrt}{w.\,r.\,t.\xspace}

\newcommand{\NN}{\mbox{\rm I$\!$N}}
\newcommand{\closu}[1]{\overline{#1}}
\newcommand{\djoko}{Djokovi\'{c} relation\xspace}

\newcommand{\subE}[1]{{#1}_{\mathcal{E}}}

\newsavebox{\Proofsym}
\savebox{\Proofsym}
{%
\begin{picture}(7,7)
\put(0,0){\framebox(6,6){}}
\end{picture}
}

\newcount\shortyear\newcount\shorthour\newcount\shortminute
\shorthour=\time\divide\shorthour by 60\shortyear=\shorthour
\multiply\shortyear by 60\shortminute=\time\advance\shortminute by
-\shortyear
\shortyear=\year\advance\shortyear by -1900

\def\zeit{\number\shorthour:\ifnum\shortminute<10 0\number\shortminute
\else\number\shortminute\fi}

\def\rightmark{{\small \today, \zeit{}}}
\begin{document}


\title{Finding all convex cuts of a plane graph in polynomial
  time\footnote{Parts of this paper have been published in a
    preliminary form in \emph{Proceedings of the 8th International
      Conference on Algorithms and Complexity (CIAC'13)~\cite{Glantz2013c}.}}}

\author{Roland Glantz \and Henning Meyerhenke}

\author{{\sc Roland Glantz\thanks{Homepage: \url{http://parco.iti.kit.edu/glantz/}}, \, Henning Meyerhenke\thanks{Homepage: \url{http://parco.iti.kit.edu/henningm/}}} \\ 
{\small\em Institute of Theoretical Informatics} \\ 
{\small \em Karlsruhe Institute of Technology (KIT)} \\
{\small\em Karlsruhe, Germany}}

\maketitle

\begin{abstract}
Convexity is a notion that has been defined for subsets of $\RR^n$ and
for subsets of general graphs. A convex cut of a graph $G=(V, E)$ is a
$2$-partition $V_1 \dot{\cup} V_2=V$ such that both $V_1$ and $V_2$
are convex, \ie shortest paths between vertices in $V_i$ never leave
$V_i$, $i \in \{1, 2\}$. Finding convex cuts is $\mathcal{NP}$-hard for
general graphs. To characterize convex cuts, we employ the \djoko,
a reflexive and symmetric relation on the edges of a graph that is
based on shortest paths between the edges' end vertices.

It is known
for a long time that, if $G$ is bipartite and the \djoko is transitive
on $G$, \ie $G$ is a partial cube, then the cut-sets of $G$'s convex
cuts are precisely the equivalence classes of the \djoko. In
particular, any edge of $G$ is contained in the cut-set of exactly one
convex cut. We first characterize a class of plane graphs that we call {\em
  well-arranged}. These graphs are not necessarily partial cubes, but
any edge of a well-arranged graph is contained in the cut-set(s) of at
least one convex cut. Moreover, the cuts can be embedded into the
plane such that they form an arrangement of pseudolines, or a slight
generalization thereof. Although a well-arranged graph $G$ is not
necessarily a partial cube, there always exists a partial cube that
contains a subdivision of $G$.

We also present an algorithm that uses the \djoko for computing all
convex cuts of a (not necessarily plane) bipartite graph in
$\bigO(|E|^3)$ time. Specifically, a cut-set is the cut-set of a
convex cut if and only if the \djoko holds for any pair of edges in
the cut-set.

We then characterize the cut-sets of the convex cuts of a general
graph $H$ using two binary relations on edges: (i) the \djoko on the
edges of a subdivision of $H$, where any edge of $H$ is subdivided
into exactly two edges and (ii) a relation on the edges of $H$ itself
that is not the \djoko. Finally, we use this characterization to present
the first algorithm for finding all convex cuts of a plane graph in 
polynomial time.\\

\end{abstract}

%
%
%
\section{Introduction}
A {\em convex $k$-partition} of an undirected graph $G=(V,E)$ is a
partition $(V_1, \dots, V_k)$ of $V$ such that the subgraphs of $G$
induced by $V_1, \dots V_k$ are all convex. A convex subgraph of $G$,
in turn, is a subgraph $S$ of $G$ such that for any pair of vertices
$v,w$ in $S$ all shortest paths from $v$ to $w$ in $G$ are fully
contained in $S$. The vertex set of a convex subgraph is called {\em
  convex set}.

A {\em convex cut} of $G$ is a convex $2$-partition of $G$. If $G$ has
a convex $k$-partition, then $G$ is said to be {\em
  $k$-convex}. Artigas \etal~\cite{Artigas20111968} showed that, for a
given $k \geq 2$, it is $\mathcal{NP}$-complete to decide whether a
(general) graph is $k$-convex. Moreover, given a bipartite graph
$G=(V,E)$ and an integer $l < \vert V \vert$, it is
$\mathcal{NP}$-complete to decide whether there exists a convex set
with at least $l$ vertices~\cite{DouradoPRS12convexity}.

There exists a different notion of convexity for plane graphs. A plane
graph is called convex if all of its faces are convex polygons. This
second notion is different and {\em not} object of our investigation.
The notion of convexity in acyclic directed graphs, motivated by
embedded processor technology, is also
different~\cite{Balister2009509}. There, a subgraph $S$ is called
convex if there is no {\em directed} path between any pair $v, w$ in
$S$ that leaves $S$. In addition to being directed, these paths do not
have to be {\em shortest} paths as in our case.

Applications that potentially benefit from convex cuts include
data-parallel numerical simulations on graphs. Here the graph is
partitioned into parts that have nearly the same number of
vertices~\cite{bichot2011graph,BulucMSSS13recent}. For some linear solvers used in these
simulations, the {\em shape} of the parts, in particular short
boundaries, small aspect ratios, but also connectedness and smooth
boundaries, plays a significant role~\cite{MeyerhenkeMS09new}. Convex
subgraphs of meshes typically admit these properties. Another example
is the preprocessing of road networks for shortest path queries by
partitioning according to natural
cuts~\cite{DBLP:conf/ipps/DellingGRW11}. The definition of a natural
cut is not as strict as that of a convex cut, but they have a related
motivation.

Due to the importance of graph partitions in theory and
practice~\cite{bichot2011graph}, it is natural to ask whether the time
complexity of finding convex cuts is polynomial for certain types of
input graphs. In this paper, we will see that polynomial-time
algorithms exist for a sub-class of plane graphs and for bipartite
graphs. Specifically, a cut-set is the cut-set of a convex cut of a
bipartite graph if and only if the \djoko holds for any pair of edges
in the cut-set.

We also characterize the cut-sets of the convex cuts of a general
graph $H$ in terms of two binary relations, each on a different kind of
edges: the edges of a subdivision of $H$, where any edge of $H$
is subdivided into two edges and (ii) the edges of $H$ itself. The
relation on the first kind of edges is the \djoko (see
Section~\ref{sec:prelim}), and the relation on the second kind of
edges, denoted by $\tau$, is such that $e~\tau~f$ iff the distance
between any end vertex of $e$ to any end vertex of $f$ is the same.

\subsection{Related work}
\label{sub:rel-work}
%
Artigas \etal~\cite{Artigas20111968} show that every connected chordal
graph $G=(V,E)$ is $k$-convex, for $1 \leq k \leq |V|$. They also
establish conditions on $|V|$ and $p$ to decide whether the $p$th
power of a cycle is $k$-convex. Moreover, they present a linear-time
algorithm that decides whether a cograph is $k$-convex.

Our methods for characterizing and finding convex cuts of a plane
graph $G$ are motivated by the work in Chepoi \etal~\cite{Chepoi97a}
who defined alternating cuts and specified conditions under which
alternating cuts are convex cuts and vice versa. Our approach is more
myopic, though. We call a face $F$ of $G$ even [odd] if $\vert E(F)
\vert$ is even, and an alternating path is one that cuts an even face
$F$ such that the number of vertices in the two parts of $E(F)$ is
equal. In an odd face the alternating path makes a slight left or
right turn so that the number of vertices in the two parts of $E(F)$
differ by one. As in~\cite{Chepoi97a}, when following an alternating
path through the faces of $G$, a left [right] turn must be compensated
by a right [left] turn as soon as this is possible.

Plane graphs usually have alternating cuts that are not convex and
convex cuts that are not alternating. Proposition 2
in~\cite{Chepoi97a} characterizes the set of plane graphs for which
the alternating cuts coincide with the convex cuts in terms of a
condition on the boundary of {\em any} alternating cut. In this paper
we represent the alternating cuts as plane curves that we call
embedded alternating paths (EAPs)---an EAP partitions $G$ exactly like
the alternating cut it represents. In contrast to~\cite{Chepoi97a},
however, we focus on the {\em intersections} of the EAPs (\ie
alternating cuts).

If any pair of EAPs intersects at most once, we have a slight
generalization of what is known as an {\em arrangement of
  pseudolines}. The latter arise in discrete geometry, computational
geometry, and in the theory of matroids~\cite{Bjoerner99a}. Duals of
arrangements of pseudolines are known to be partial cubes (see
Section~\ref{sec:prelim}), a fact that has been applied to graphs
before by Eppstein~\cite{Eppstein2006a}, for example. For basics on
partial cubes we rely on Ovchinnikov's
survey~\cite{Ovchinnikov2008a}. The following basic fact about partial
cubes is crucial for our method to find convex cuts: partial cubes are
precisely the graphs that are bipartite and on which the
\djoko~\cite{Djokovic73a} (defined in Section~\ref{sec:prelim}) is
transitive. For a characterization of planar partial cubes see
Peterin~\cite{Peterin2008a}.
%
\subsection{Paper outline and contribution}
%
In Section~\ref{sec:ALT_EMBED} we first represent (myopic versions of)
the alternating cuts of a plane graph $G=(V,E)$, as defined
in~\cite{Chepoi97a}, by EAPs. The main work here is on the proper
embedding. We then study the case of $G$ being {\em well-arranged}, as
we call it, \ie the case in which the EAPs form an arrangement of
pseudolines, or a slight generalization thereof. We show that the dual
$\subE{G}$ of such an arrangement is a partial cube and reveal a
one-to-one correspondence between the EAPs of $G$ and the convex cuts
of $\subE{G}$. Specifically, the edges of $\subE{G}$ intersected by an
EAP form the cut-set of a convex cut of $\subE{G}$. Conversely, the
cut-set of any convex cut of $\subE{G}$ is the set of edges
intersected by a unique EAP of $G$. From (i) the one-to-one
correspondence between the EAPs of $G$ and the convex cuts of
$\subE{G}$ and (ii) the construction of $\subE{G}$ we derive that the
EAPs also define convex cuts of $G$.


In Section~\ref{sec:bipartite} we specify an $\bigO(|E|^3)$-time
algorithm to find all convex cuts of a not necessarily plane bipartite
graph. The fact that we can compute all convex cuts in bipartite
graphs in polynomial time is no contradiction to the
$\mathcal{NP}$-completeness of the decision problem whether the
largest convex set in a bipartite graph has a certain
size~\cite{DouradoPRS12convexity}. Indeed, for a cut to be convex,
{\em both} subgraphs have to be convex, whereas the complement of a
convex set is not required to be a convex set itself.

In Section~\ref{sec:general} we characterize the cut-sets of the
convex cuts of a general graph $H$ in terms of the \djoko and
$\tau$. The results of Section~\ref{sec:general} are then used in
Section~\ref{sec:general_plane} to derive new necessary conditions for
convexity of cuts of a {\em plane} graph $G$. As in the case of
well-arranged graphs, we iteratively proceed from an edge on the
boundary of a face $F$ of $G$ to another edge on the boundary
''opposite`` of $F$. This time, however, ``opposite'' is with respect
to the \djoko on a subdivision of $G$. Thus we arrive at a polynomial-time
algorithm that finds all convex cuts of $G$. We correct an error in a
preliminary version~\cite{Glantz2013c} of this paper. The running time
is now $\bigO(\vert V \vert^7)$ instead of $\bigO(\vert V \vert^3)$.

%
%
%
\section{Preliminaries}
\label{sec:prelim}
\label{sub:notation}
Unless stated otherwise, $G = (V, E)$ is a finite, undirected,
unweighted, and two-connected graph that is free of
self-loops. Two-connectedness is not a limitation for the problem of
finding convex cuts because a convex cut cannot cut through more than
one block of $G$, and self-loops have no impact on the convex
cuts. For $e \in E$ with end points $u,v$ ($u \not= v$) we sometimes
write $e=\{u,v\}$ even when $e$ is not necessarily determined by $u$
and $v$ due to parallel edges. We use the term {\em path} in the
general sense: a path does not have to simple, and it can be a cycle.

If $G$ is plane, $V$ is a set of points in $\RR^2$, and $E$ is a set
of plane curves that intersect only at their end points which, in
turn, make up $V$. The unbounded face of $G$ is denoted by
$F_{\infty}$. For a face $F$ of $G$, we write $E(F)$ for the set
of edges that bound $F$. Our definitions and results on plane graphs
are invariant to topological isomorphism~\cite{Diestel2006a} which, in
conjunction with two-connectedness, is equivalent to combinatorial
isomorphism~\cite{Diestel2006a}. Since any plane graph is
combinatorially isomorphic to a plane graph whose edges are line
segments~\cite{Diestel2006a}, we can always resort to the case of
straight edges without loss of generality. We do so especially in our
illustrations.

We denote the standard metric on $G$ by $d_G(\cdot,\cdot)$. In this
metric the distance between $u,v \in V$ amounts to the number of edges
on a shortest path from $u$ to $v$.
A subgraph $S=(V_S,E_S)$ of a (not necessarily plane) graph $H$ is an
{\em isometric} subgraph of $H$ if $d_S(u,v)=d_H(u,v)$ for all $u,v
\in V_S$.

Following Djokovi\'{c}~\cite{Djokovic73a} and using Ovchinnikov's 
notation~\cite{Ovchinnikov2008a}, we set
\begin{equation}
W_{xy}=\{w \in V : d_{G}(w,x)<d_{G}(w,y)\} \quad \forall\{x,y\}\in E.
\end{equation}

Let $e=\{x,y\}$ and $f=\{u,v\}$ be two edges of $G$. The \djoko
$\theta$ on $G$'s edges is defined as follows:
\begin{equation}
e~\theta~f\Leftrightarrow f \mbox{~has one end vertex in $W_{xy}$ and
  one in $W_{yx}$}.
\end{equation}

The \djoko is reflexive, symmetric~\cite{Wilkeit90a}, but not
necessarily transitive. The vertex set $V$ of $G$ is partitioned by
$W_{ab}$ and $W_{ba}$ if and only if $G$ is bipartite.

A {\em partial cube} $G_q=(V_{q},E_{q})$ is an isometric subgraph of a
hypercube. Interested readers find more details on partial cubes in
Ovchinnikov's survey~\cite{Ovchinnikov2008a}; we state a few important results
for the sake of self-containment.
Partial cubes and $\theta$ are related in that a graph is a partial
cube if and only if it is bipartite and $\theta$ is
transitive. Thus, $\theta$ is an equivalence
hypercube. For a survey on partial cubes see
Ovchinnikov~\cite{Ovchinnikov2008a}.  Partial cubes and $\theta$ are
related in that a graph is a partial cube if and only if it is
bipartite and $\theta$ is transitive. Thus, $\theta$ is an equivalence
relation on $E_q$, and the equivalence classes are cut-sets of
$G_q$. Moreover, the cuts defined by these cut-sets are precisely the
convex cuts of $G_q$. If $(V_q^1,V_q^2)$ is a
convex cut, the (convex) subgraphs induced by $V_q^1$ and $V_q^2$ are
called \emph{semicubes}.  If $\theta$ gives rise to $k$ equivalence
classes $E_q^{1},\dots E_q^{k}$, and thus $k$ pairs
$(S_{a}^{i},S_{b}^{i})$ of semicubes, where the ordering of the
semicubes in the pair is arbitrary, one can derive a \emph{Hamming
  labeling} $b:V_q \mapsto \{0,1\}^{k}$ by setting

\begin{equation}
b(v)_{i}=\left\{\begin{array}{ll}
0 & \quad \mbox{if $v\in S_{a}^{i}$}\\
1 & \quad \mbox{if $v\in S_{b}^{i}$}
\end{array}\right.
\end{equation}

In particular, $d_{G_q}(u,v)$ amounts to the Hamming distance between
$b(u)$ and $b(v)$ for all $u,v \in V_q$. This is a consequence of the
fact that the corners of a hypercube have such a labeling and that
$G_q$ is an isometric subgraph of a hypercube.

%
%
%
\section{Partial cubes from embedding alternating paths}
\label{sec:ALT_EMBED}
In Section~\ref{sub:alternating} we define a multiset of (not yet
embedded) alternating paths of a graph $G$. Section~\ref{sub:embed} is
devoted to embedding the alternating paths into $\RR^2$ and to the
definition of well-arranged graphs. In Section~\ref{sub:ell_one} we
study the dual of an embedding of alternating paths and show that it
is a partial cube whenever $G$ is well-arranged.
 %
%
\subsection{Alternating paths}
\label{sub:alternating}
%
Intuitively, an embedded alternating path $P$ runs through a face $F$
of $G$ such that the edges through which $P$ enters and leaves $F$ are
opposite---or nearly opposite because, if $\vert E(F) \vert$ is odd,
there is no opposite edge, and $P$ has to make a slight turn to the
left or to the right. The exact definitions leading up to (not yet
embedded) alternating paths are as follows.
\begin{definition}[Opposite edges, left, right, unique opposite edge]
Let $F \not= F_{\infty}$ be a face of $G$, and let $e, f \in
E(F)$. Then $e$ and $f$ are called {\em opposite edges} of $F$ if the
lengths of the two paths induced by $E(F) \setminus \{e,f\}$ differ by
at most one. If the two paths have different lengths, $f$ is called
the {\em left [right] opposite edge} of $e$ if starting on $e$ and
running clockwise around $F$, the shorter [longer] path comes
first. Otherwise, $e$ and $f$ are called {\em unique opposite edges}.
\label{def:opposite_edges}
\end{definition}
\begin{definition}[Alternating path graph $A(G)=(V_A,E_A)$]
The {\em alternating path graph} $A(G)=(V_A,E_A)$ of $G=(V,E)$ is the
(non-plane) graph with $V_A=E$ and $E_A$ consisting of all two-element
subsets $\{e,f\}$ of $E$ such that $e$ and $f$ are opposite edges of
some face $F \not= F_{\infty}$.
\label{def:alternating_path_graph}
\end{definition}
The alternating path graph defined above will provide the edges for
the multiset of alternating paths defined next. We resort to a {\em
  multiset} for the sake of uniformity, i.e., to ensure that any edge
of $G$ is contained in exactly two alternating paths (see
Figure~\ref{fig:saws}a).

\begin{figure}
\begin{center}
(a) \includegraphics[width=0.40\columnwidth]{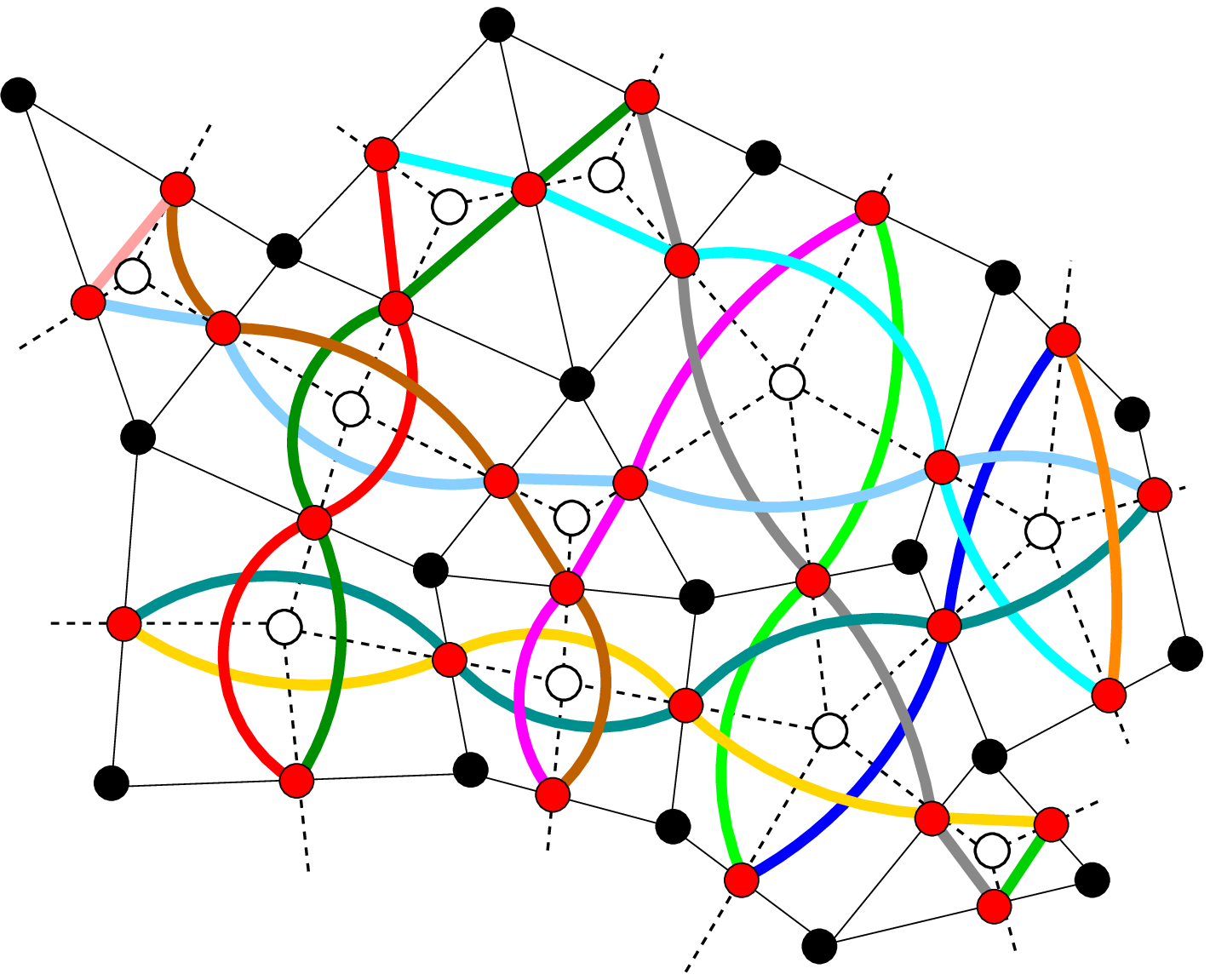}\qquad
(b) \includegraphics[width=0.40\columnwidth]{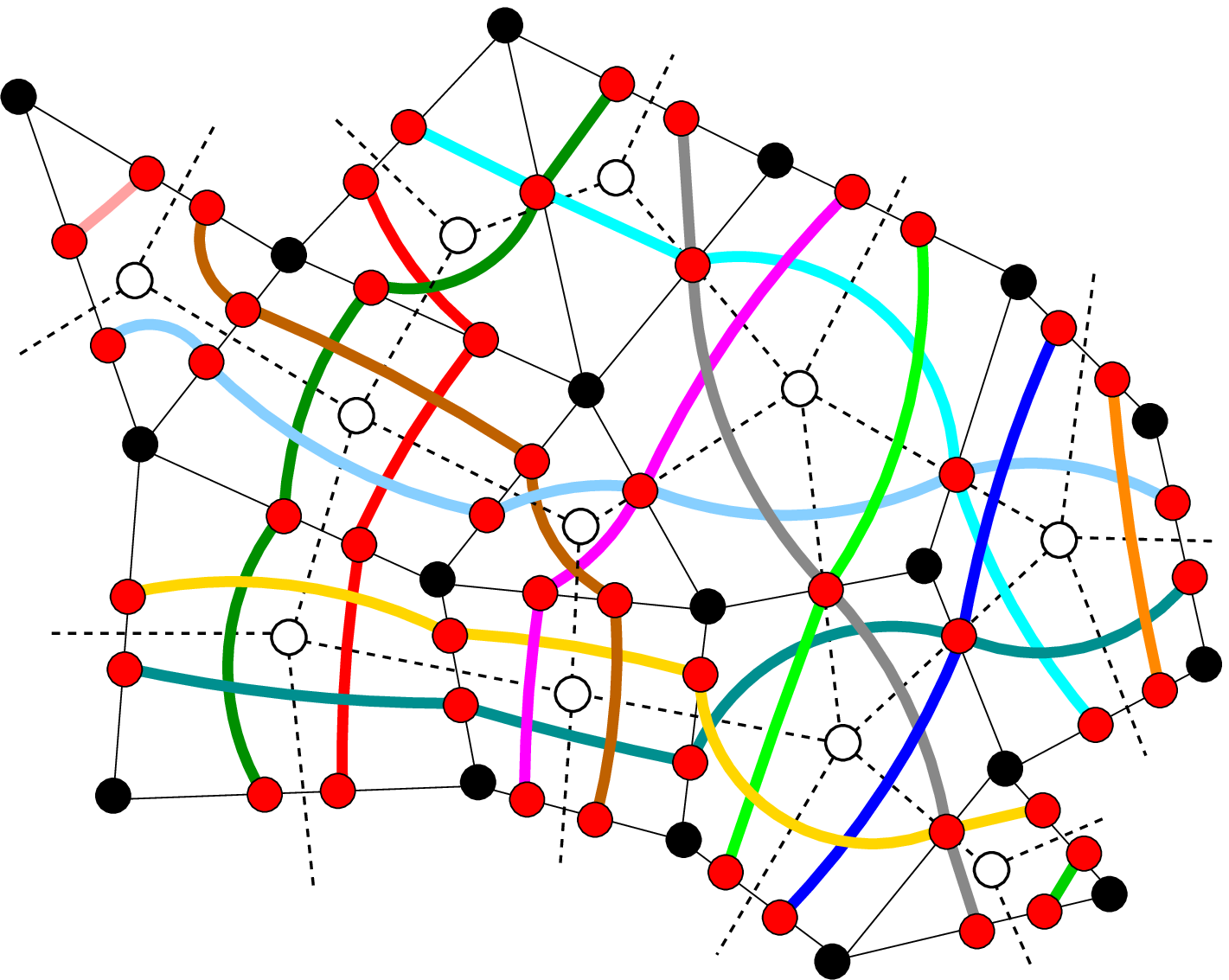}
\caption{Primal graph: black vertices, thin solid edges. Dual graph:
  white vertices, dashed edges. (a) Multiset $\mathcal{P}(G)$ of
  alternating paths: Red vertices, thick solid lines. The paths in
  $\mathcal{P}(G)$ are colored. In this ad-hoc drawing the two
  alternating paths that share a vertex, \ie an edge of $G$, go
  through the same (red) point on the edge. (b) Collection
  $\mathcal{E}(G)$ of EAPs: Red vertices, thick solid colored lines.}
\label{fig:saws}
\end{center}
\end{figure}

\begin{definition}[(Multiset $\mathcal{P}(G)$ of) alternating paths in $A(G)$]
A {\em maximal} path $P=(v_A^1, e_A^1, v_A^2, \dots e_A^{n-1},
v_A^n)$ in $A(G)=(V_A,E_A)$ is called {\em alternating} if
\begin{itemize}
\item[(i)] $v_A^i$ and $v_A^{i+1}$ are opposite for all $1 \leq i \leq n-1$ and
\item[(ii)] if $v_A^{i+1}$ is the left [right] opposite of $v_A^i$,
  and if $j$ is the minimal index greater than $i$ such that $v_A^j$
  and $v_A^{j+1}$ are not unique opposites (and $j$ exists at all),
  then $v_A^{j+1}$ is the right [left] opposite of $v_A^j$.
\end{itemize}

We (arbitrarily) select one path from each pair formed by an
alternating path $P$ and the reverse of $P$. The multiset
$\mathcal{P}(G)$ contains all selected paths: the
multiplicity of $P$ in $\mathcal{P}(G)$ is two if $v_A^{i+1}$ is a
unique opposite of $v_A^i$ for all $1 \leq i \leq n-1$, and one
otherwise.
\label{def:multiset}
\end{definition}

\subsection{Embedding of alternating paths}
\label{sub:embed}
%
In this section we embed the alternating paths of a plane graph $G$
into $\RR^2$. We may assume that the edges of $G$ are straight line
segments (see Section~\ref{sec:prelim}). An edge $\{e,f\}$ of an
alternating path turns into a non-self-intersecting plane curve with
one end point on $e$ and the other end point on $f$. An alternating
path with multiplicity $m \in \{1,2\}$ gives rise to $m$ embedded
paths. Visually, we go from Figure~\ref{fig:saws}a to
Figure~\ref{fig:saws}b.

Note that any edge $e$ of $G$ is contained in exactly two alternating
paths. For any $e$ that separates two bounded odd faces we
predetermine a point $s$ on $e$'s interior and require that both
alternating paths containing $e$ must run through $s$. If $e$ does not
separate two odd faces, we predetermine two points, $s_1$ and $s_2$,
on $e$'s interior and require that one alternating path runs through
$s_1$ and the other one runs through $s_2$. We refer to $s$, $s_1$ and
$s_2$ as {\em slots} of $e$. If $P=(v_A^1, e_A^1, v_A^2, \dots
e_A^{n-1}, v_A^n)$ is an alternating path, let $F_i = F_i(P)$ be the
$i$th (bounded) face of $G$ that will be traversed by embedded $P$,
\ie the (unique) face with $v_A^i, v_A^{i+1} \in E(F_i)$. Since we
required that $G$ is two-connected, we have $v_A^i \neq
v_A^{i+1}$. Thus, if $v_A^i$ has two slots, there exists a
well-defined left and right slot from the perspective of standing on
$v_A^i$ and looking into $F_i$, $1 \leq i < n$. Finally, left and
right on $v_A^n$ is from the perspective of looking into $F_{\infty}$.

The overall scenario is that we embed the alternating paths one
after the other, where the order of the paths is arbitrary. The
following rules for an individual path $P$ then determine which slots
are occupied by which alternating paths. For an example of slot choice
see Figure~\ref{fig:crossings}a,b. The variable $a(P)$ is zero if and
only if $P$ makes no left and no right turn. Otherwise, $a(P)$
indicates the preference for the next slot at any time.

\begin{enumerate}
\item If $P$ has no left turn and no right turn, set $a(P)$ to
  $0$. Otherwise, if the first turn of $P$ is a left [right] turn, set
  $a(P)$ to $l$ [$r$].
\item Let $s_l$ [$s_r$] be the left [right] slot on $v_A^1$. If $a(P)
  = 0$, choose a vacant slot (arbitrarily if both slots are
  vacant). If $a(P) = l$ [$a(P) = r$], occupy the left [right] slot
  if that slot is still vacant. Otherwise, occupy the alternative slot
  and set $a(P) = r$ [$a(P) = l$].
\item Assume that we have found slots for $v_A^1, \dots v_A^i$.
 
\begin{itemize}
\item If $a(P) = 0$ and the slot occupied on $v_A^i$ was the left
  [right] slot, then occupy the left [right] slot on $v_A^{i+1}$.
\item If $v_A^{i+1}$ has only one slot, then occupy it (single slots can be
  occupied by two paths). If $(a(P) = l)$ [$(a(P) = r)$], then set $(a(P) =
  r)$, [$(a(P) = l)$].
\item If $v_A^{i+1}$ has two slots and $a(P) = l$ [$a(P) = l$],
  then occupy the left [right] slot, if vacant. Otherwise, occupy the
  alternative slot and set $a(P) = r$ [$a(P) = l$].
\end{itemize}
\end{enumerate}

The embedding of the alternating paths will be such that two paths
that share a point $p$ will always cross at $p$, and not just touch
(see Proposition~\ref{prop:share} and Figure~\ref{fig:crossings}c). We
are not interested in the exact course of the embedded alternating
paths (EAPs), but only in their {\em intersection pattern}, \ie
whether certain pairs of EAPs cross in a certain face or on a certain
edge. The intersection pattern is not going to be unique, but our
central definition, \ie that of a well-arranged graph, will be
invariant to ambiguities of intersection patterns (see
Proposition~\ref{prop:egal}).

Next we formulate rules for embedding a single edge $\{e, f\}$ of an
alternating path into $\RR^2$. If $F$ is the unique face of $G$ with
$\{e, f\} \subset E(F)$, we embed $\{e, f\}$ into $\closu{F} = F \cup
E(F)$. To this end, we first represent $\closu{F}$ by a {\em
  regular} filled polygon $\closu{F_r}$ with the same number of
sides. We then embed $\{e, f\}$ into $\closu{F_r}$ as a line segment
$L$ between two points on the sides of $\closu{F_r}$. Due to the
Jordan-Sch\"{o}nflies theorem~\cite{Thomassen92a}, there exists a
homeomorphism $\hslash:\closu{F_r} \mapsto \closu{F}$. The embedding
of $\{e, f\}$ into $\closu{F}$ is then given by $\hslash(L)$. Since
$\hslash$ is a homeomorphism, the intersection pattern of the line
segments in $\closu{F}$ coincides with that in $\closu{F_r}$ (see
Figures~\ref{fig:intersection_patterns}a and~\ref{fig:homeo}a).

\begin{figure}[t]
\begin{center}
(a) \includegraphics[height=40mm]{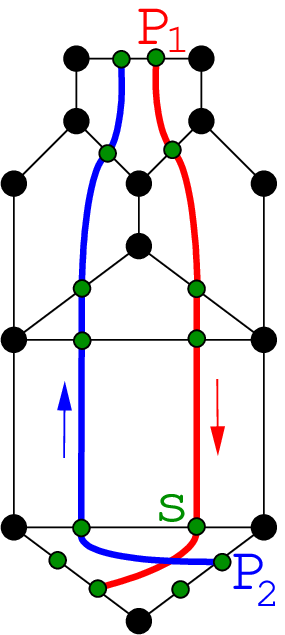}\qquad
(b) \includegraphics[height=40mm]{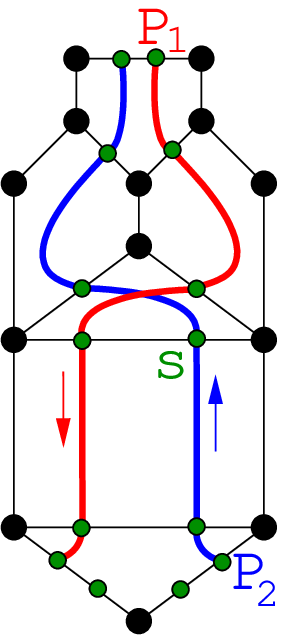}\qquad
(c) \includegraphics[height=35mm]{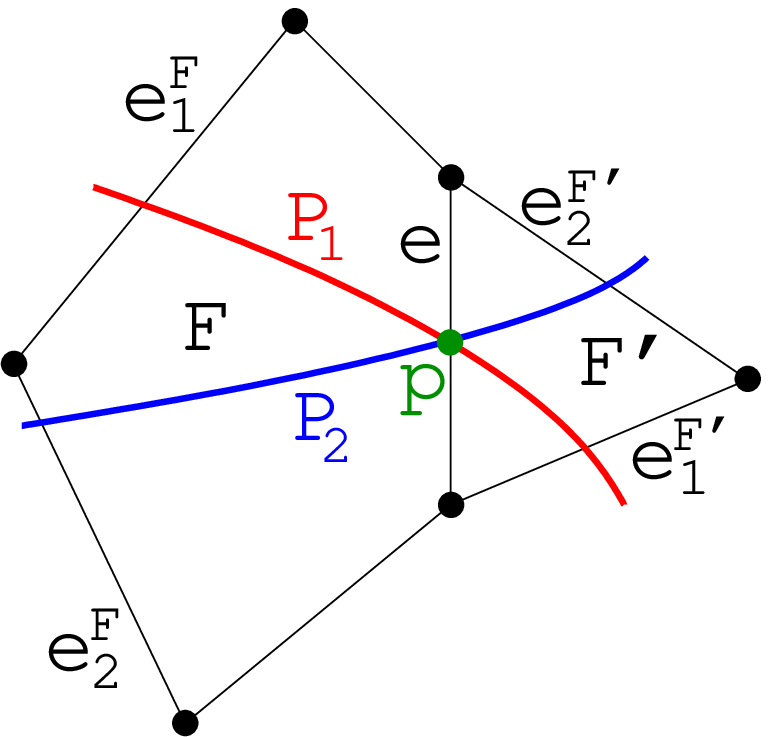}
\caption{(a,b) Slots are colored in green, and slot conflicts occur at
  $s$.(a) $P_1$ picked the slots before $P_2$. (b) $P_2$ picked the
  slots before $P_1$. (c) Illustration to proof of
  Proposition~\ref{prop:share}.}
\label{fig:crossings}
\end{center}
\end{figure}

\begin{figure}[b]
\begin{center}
(a) \includegraphics[height=40mm]{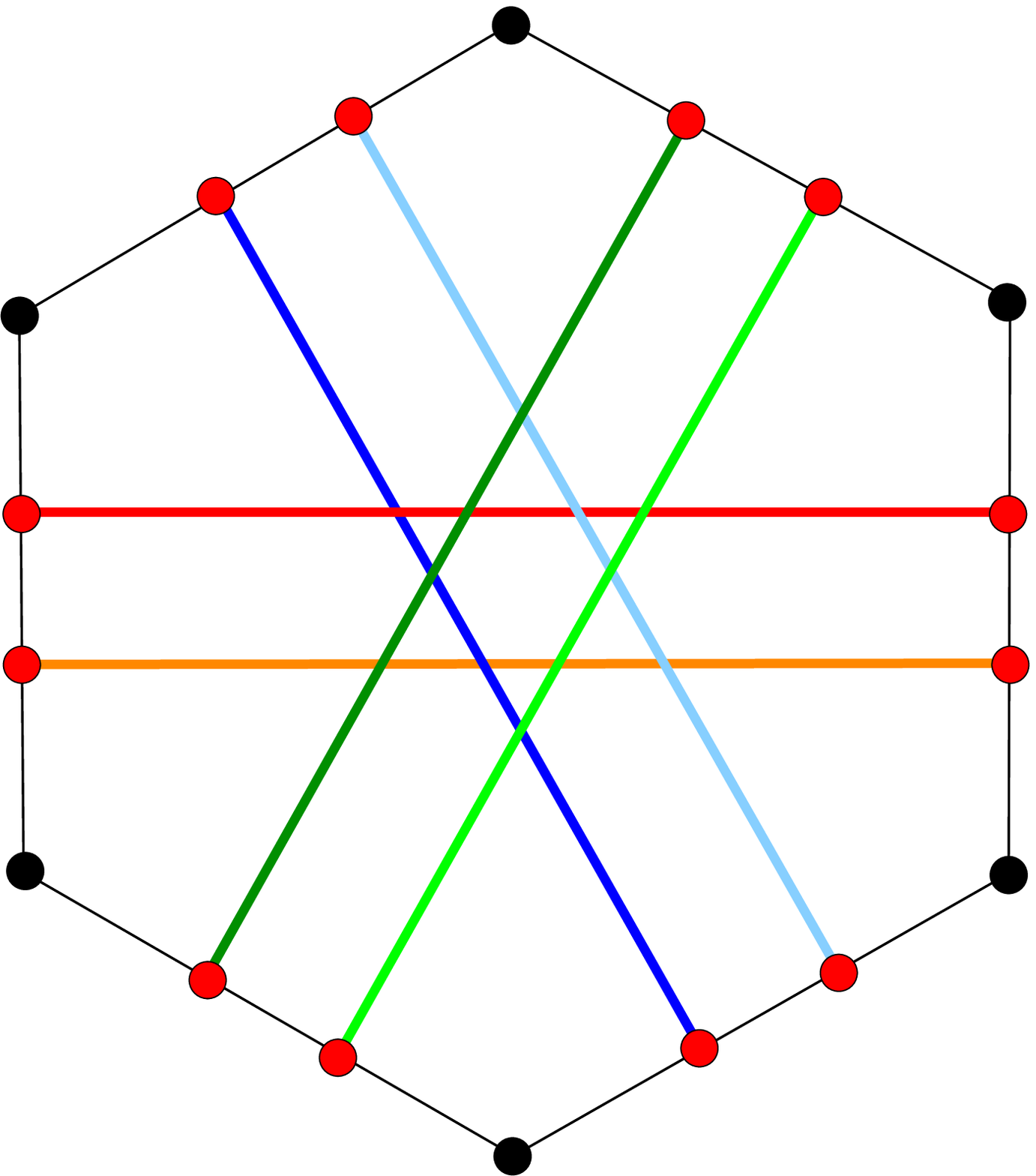}
(b) \includegraphics[height=40mm]{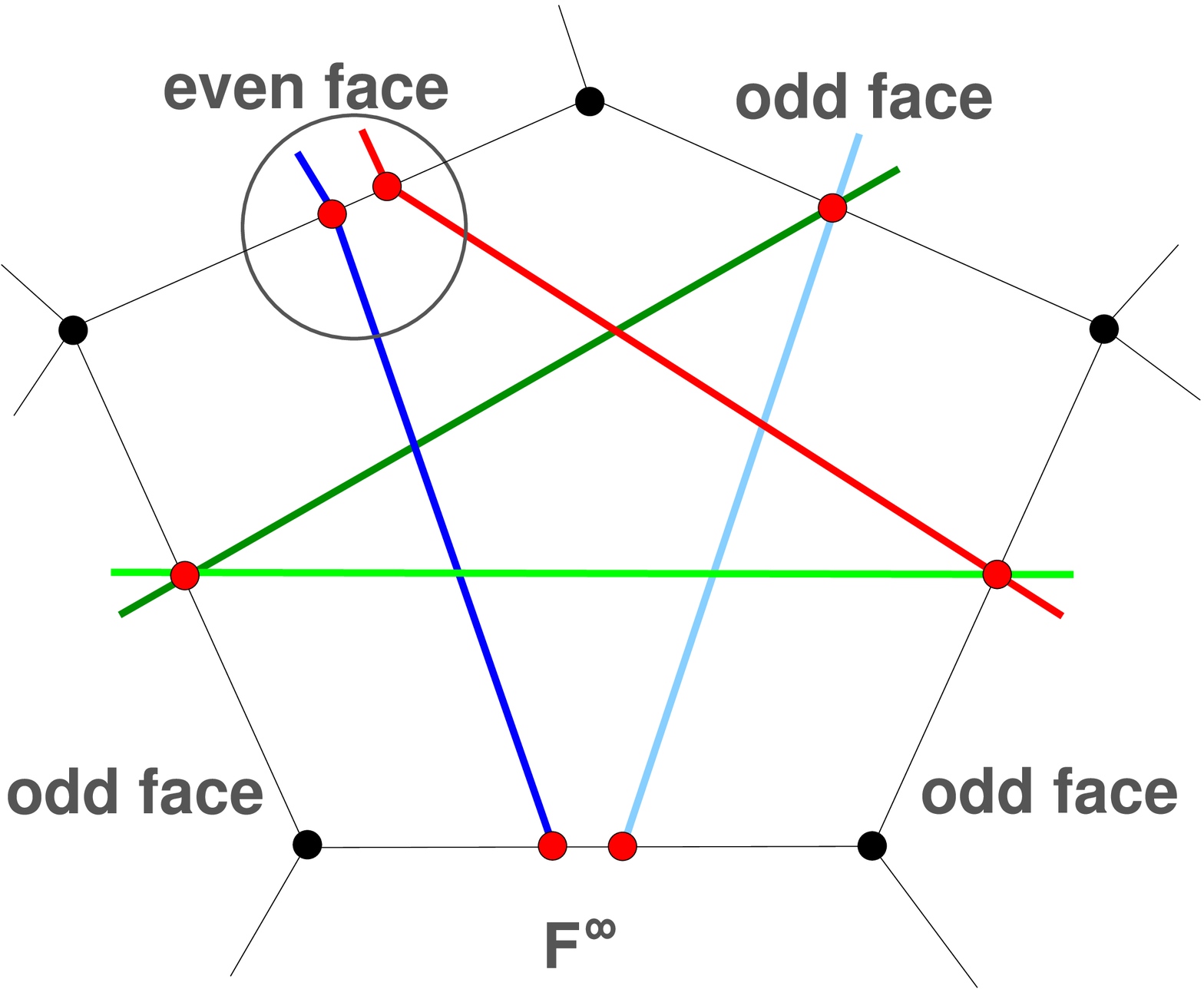}
(c) \includegraphics[height=40mm]{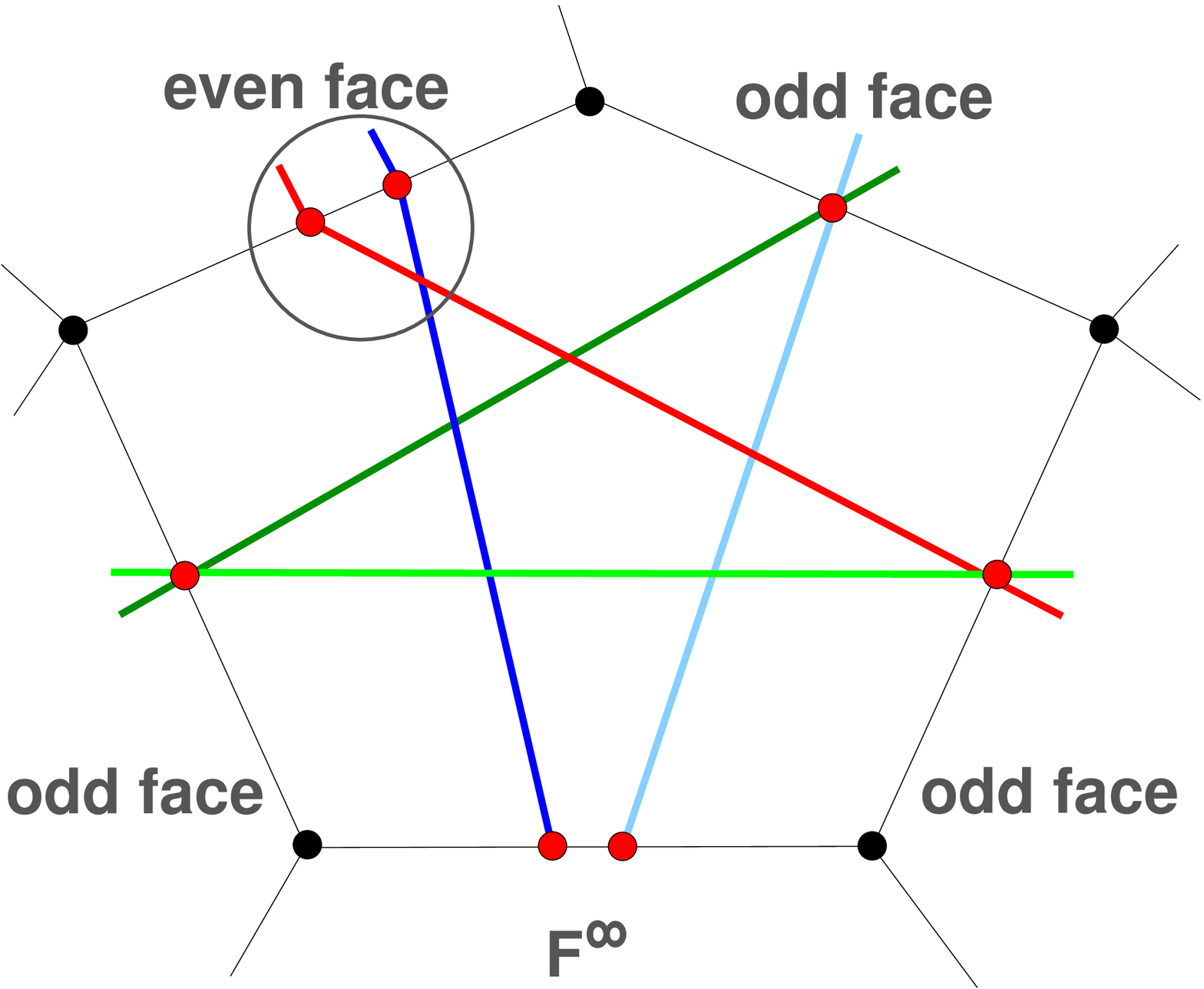}
\caption{(a) Intersection pattern in a regular hexagon (b,c) Two
  intersection patterns in a regular pentagon (see gray circle for
  difference). In (b) we have no slot conflict on the upper left edge
  of the hexagon, and in (c) we have a slot conflict on the upper left
  edge.}
\label{fig:intersection_patterns}
\end{center}
\end{figure}

\begin{figure}[t]
\begin{center}
(a) \includegraphics[height=35mm]{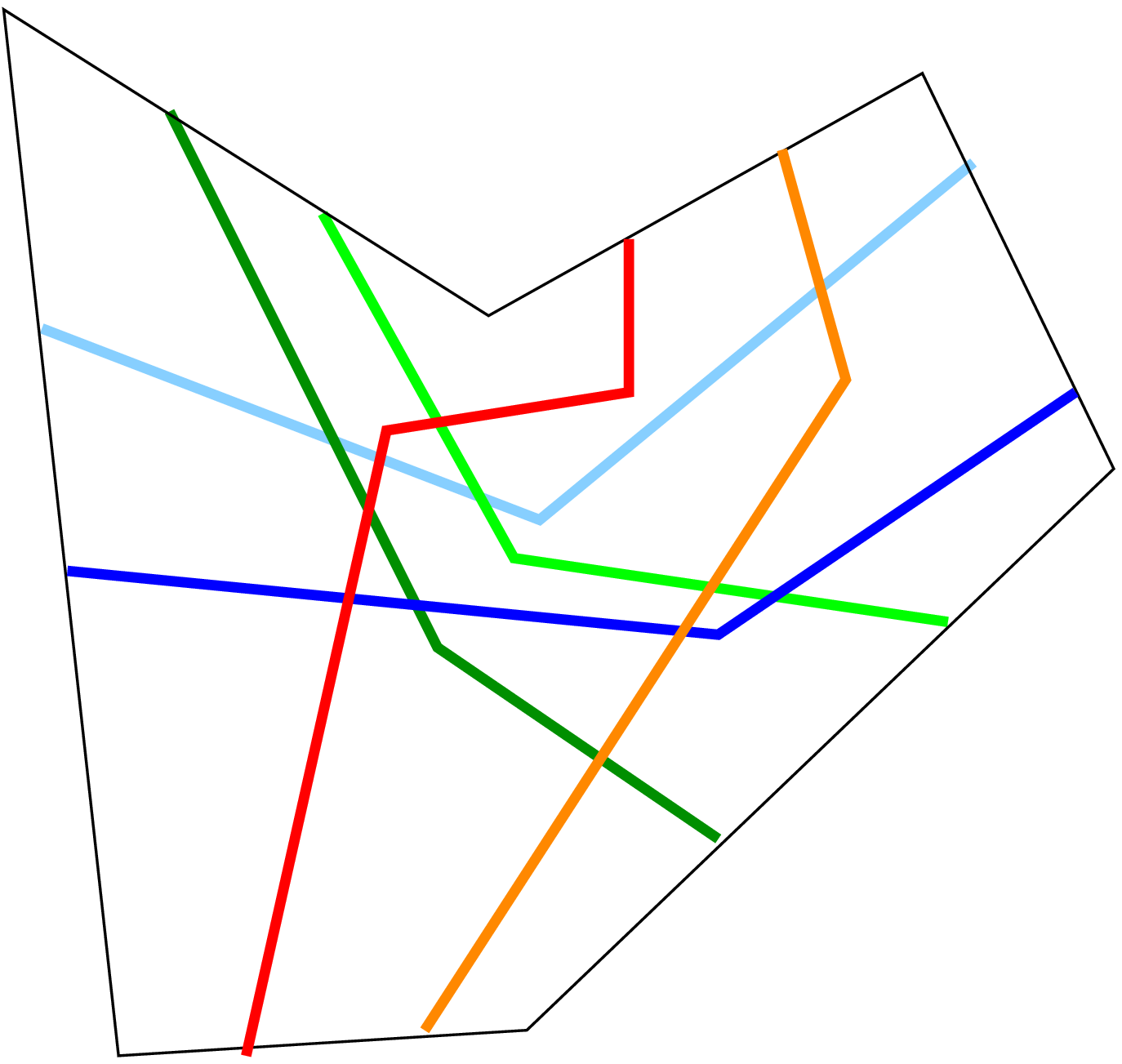}
(b) \includegraphics[height=35mm]{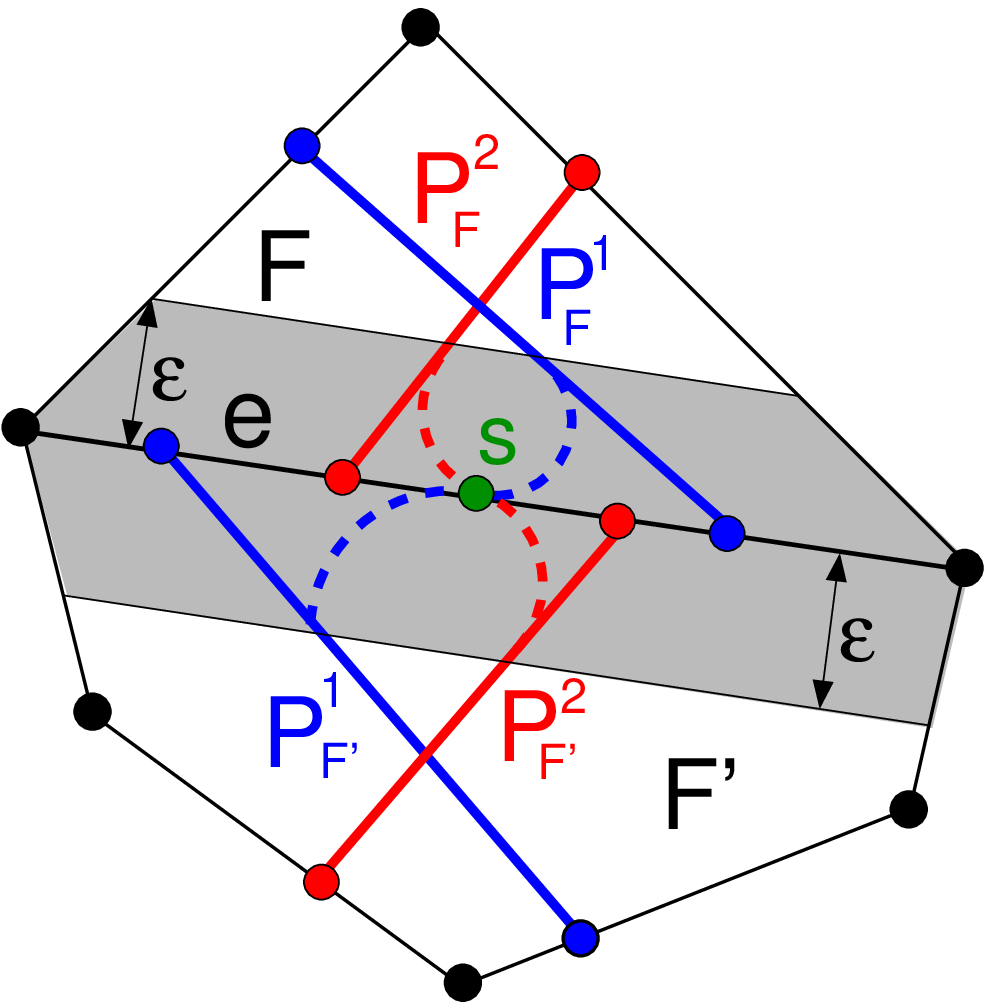}
(c) \includegraphics[height=35mm]{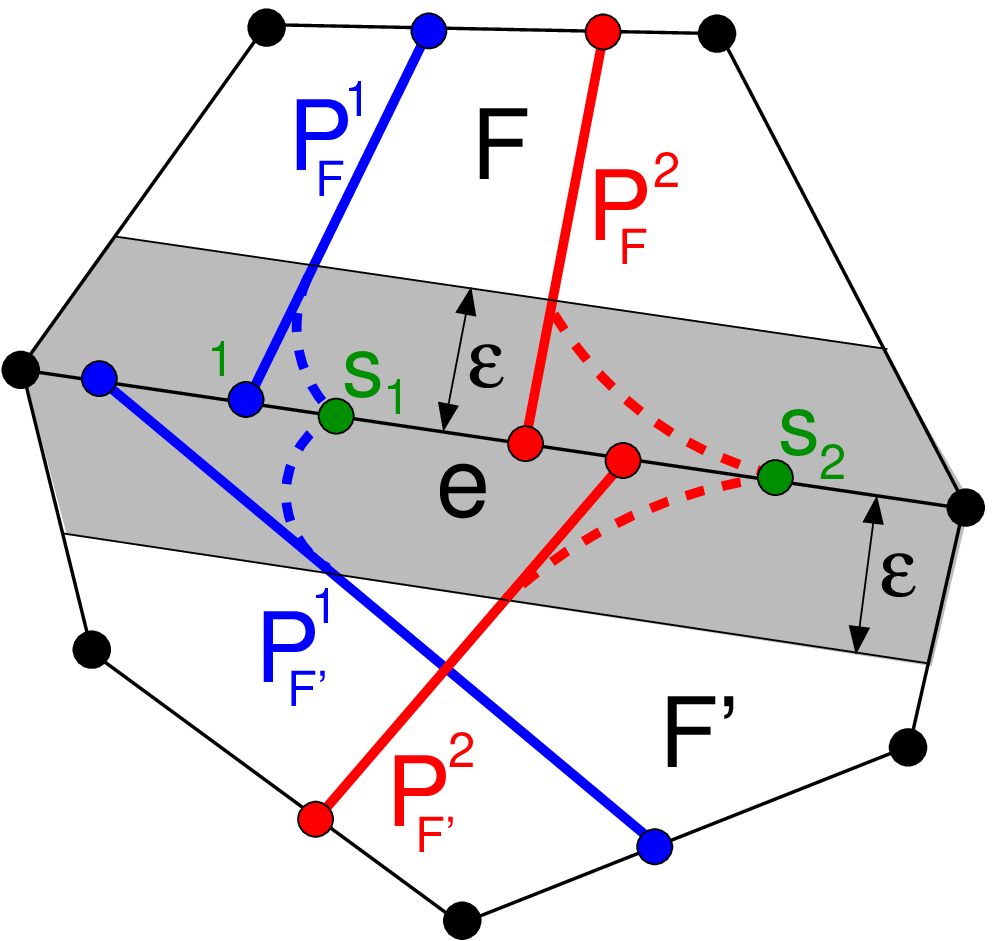}
\caption{(a) Hexagonal face with the same intersection pattern as in
  Figure~\ref{fig:intersection_patterns}a. (b,c) Bending of
  alternating paths is indicated by the dashed colored lines. For the
  notation see the text. (b) Bending toward a single slot $s$ on
  $e$. (c) Bending toward two slots $s_1, s_2$ on $e$.}
\label{fig:homeo}
\end{center}
\end{figure}

\noindent \textbf{Local embedding rules.} The local rules for
embedding an alternating path into $\closu{F_r}$ are as follows.
\begin{enumerate}
\item The part of an EAP that runs through $F_r$ is a line segment,
  and EAPs cannot coincide in $F_r$.
\item An EAP can intersect $e \in E(F_r)$ only in $e$'s relative
  interior, \ie not at $e$'s end vertices.
\item Let $F_r \not= F_{\infty}$ be an even face of $G$, let $e, f$ be
  unique opposite edges in $E(F_r)$, and let $P_1, P_2$ be the two
  alternating paths that contain the edge $\{e,f\}$ ($P_1=P_2$ if and
  only if the multiplicity of $P_1$ is two). Then the parts of
  embedded $P_1$ and $P_2$ that run through $\closu{F_r}$ must form a
  pair of distinct parallel line segments (see
  Figure~\ref{fig:intersection_patterns}a).
\item Let $F \not=F_{\infty}$ be an odd face of $G$, let $e \in
  E(F_r)$, and let $P_1, P_2$ be the two alternating paths that
  contain the vertex $e$. If $e$ also bounds an even bounded face or
  $F_{\infty}$, embedded $P_1, P_2$ must intersect $e$ at two distinct
  points (see the upper left edge of the hexagon in
  Figures~\ref{fig:intersection_patterns}b,c). If the other face is a
  bounded odd face, embedded $P_1, P_2$ must cross at a point on $e$
  (see the upper right edge of the hexagon in
  Figures~\ref{fig:intersection_patterns}b,c).
\item A slot conflict on an edge $e$ of $G$ can occur only if $e$
  separates a bounded odd face $F$ from a face that is not both
  bounded and odd. Let $\{e, f\}$, $\{e, f'\}$ be the two edges (of
  alternating paths) occupying the two slots of $e$, and let $F_r$ be
  the regular polygon that represents $F$. Then $\{e, f\}$ and $\{e,
  f'\}$ cross inside $F_r$ if and only if there is a slot conflict on
  $e$ (see Figures~\ref{fig:intersection_patterns}b,c).
\end{enumerate}

We map the embedded edges (of alternating paths) from any
$\closu{F_r}$ into $\closu{F}$ using a homeomorphism from
$\closu{F_r}$ onto $\closu{F}$. The following is about tying the loose
ends of the locally embedded paths, which all sit on edges of $G$, so
as to arrive at a {\em global} embedding of the alternating paths (see
Figures~\ref{fig:homeo}b,c).  Let $e\in E(G)$, and let $F, F'$ be the
faces of $G$ that are bounded by $e$. We have two locally embedded
paths $P_F^1, P_F^2$ in $\closu{F}$ and two locally embedded paths
$P_{F'}^1, P_{F'}^2$ in $\closu{F'}$ that all hit $e$. We bend the
four paths toward their predetermined slots. The bending operations
can be done such that the intersection patterns do not change in the
interiors of $F$ and $F'$. Indeed, recall that the paths are
homeomorphic to straight line segments. Thus, there exists $\epsilon >
0$ such that all locally embedded paths in $F$ and $F'$ other than
$P_F^1$, $P_F^2$, $P_{F'}^1$, and $P_{F'}^2$ keep a distance greater
than $\epsilon$ to $e$.  The bending, in turn, can be done such that
it affects $P_F^1$, $P_F^2$, $P_{F'}^1$, and $P_{F'}^2$ only in an
$\epsilon$-neighborhood of $e$.

\begin{proposition}
\label{prop:share}
If two EAPs share a point $p$, they cross at $p$ and not just touch.
\end{proposition}

\begin{proof}
Consider two EAPs $P_1$ and $P_2$ that share a point $p$. If $p$ sits
in a face $F$ of $G$, $P_1$ and $P_2$ cross at $p$ because (i) $P_1 =
\hslash(L_1)$ and $P_2 = \hslash(L_2)$ for a homeomorphism $\hslash:
F_r \mapsto F$, (ii) $L_1 \not= L_2$ cannot touch without crossing,
and (iii) homeomorphisms preserve crossings and non-crossings of
curves.

If $p$ sits on (the interior of an edge) $e \in E$, the two faces
separated by $e$, and denoted by $F$ and $F'$, must be finite and
odd. As illustrated in Figure~\ref{fig:crossings}c, $P_1$ [$P_2$]
enters $F$ through an edge $e_F^1$ [$e_F^2$] in $E(F) \setminus
E(F')$, runs from $F$ into $F'$ via $e$, and then leaves $F'$ via
an edge $e_{F'}^1$ [$e_{F'}^2$] in $E(F') \setminus E(F)$. Since $F$
and $F'$ are finite and odd, we have $e_F^1 \neq e_F^2$ and $e_{F'}^1
\neq e_{F'}^2$. Without loss of generality, $e$ is the left [right]
opposite of $e_F^1$ [$e_F^2$]. Then, due to item (ii) in
Definition~\ref{def:multiset}, $e_{F'}^1$ [$e_{F'}^2$] is the right
[left] opposite of $e$. Thus, before reaching point $p$ on $e$, $P_1$
is left of $P_2$, and after leaving $p$, $P_1$ is right of $P_2$. In
other words, $P_1$ and $P_2$ cross at $p$.
\end{proof}

EAPs like the ones in Figure~\ref{fig:saws}b are special in that they
form an arrangement in the following sense.

\begin{definition}[Arrangement of alternating paths]
A collection of all EAPs in $G$ is called an arrangement of embedded
alternating paths if (i) none of the EAPs crosses itself and (ii) no
pair of EAPs crosses twice.
\label{def:arrangement}
\end{definition}

We will now see that Definition~\ref{def:arrangement} does not depend
on the particular collection of EAPs, \ie that it is actually a
definition for $G$.

\begin{proposition}
If one collection of EAPs in $G$ is an arrangement of alternating
paths, then any collection of EAPs in $G$ is an arrangement of
alternating paths.
\label{prop:egal}
\end{proposition}

\begin{proof}
Definition~\ref{def:arrangement} depends only on the intersection
pattern of the EAPs. Different intersection patterns, in turn, can
only arise from different solutions of slot conflicts.

No EAP $P$ can have a slot conflict with itself, because this would
mean that $P$ would traverse a face twice, a contradiction because
then all the other EAPs that traverse the face would cross $P$ twice.

Thus, it suffices to consider slot conflicts between different
EAPs. Due to (i) $P_1$ and $P_2$ being alternating paths and (ii) the
way we assigned the slots, the intersection pattern of two EAPs $P_1$
and $P_2$ that cross edges of $G$ from the same side is unique. If,
however, an edge $e$ of $G$ is crossed by $P_1$ in one direction, and
by $P_2$ in the opposite direction, and if $e \in E(F)$ for a bounded
odd face $F$, the intersection pattern of $P_1$ and $P_2$ depends on
whether $P_1$ or $P_2$ was embedded first. This case is illustrated in
Figures~\ref{fig:crossings}a,b.

It remains to show that the above ambiguity in intersection patterns
does not turn an arrangement into a non-arrangement or vice
versa. Indeed, if $F$ is the only bounded odd face traversed by $P_1$
and $P_2$, then $P_1$ and $P_2$ do not cross in $F$, anyway.

Now assume that there exists another bounded odd face $\hat{F}$
traversed by $P_1$ and $P_2$. Examples for $F$, $\hat{F}$ are the
lower and central triangular faces in
Figures~\ref{fig:crossings}a,b. We denote by $P_1^*$ the reverse of
$P_1$. Without loss of generality we assume that $P_1^*$ [$P_2$] turns
left [right] on $F$. Due to the slot conflict on $E(F)$ (resulting in
the crossing of $P_1^*$ and $P_2$ in $F$), $P_1^*$ then runs on the
right side of $P_2^*$. Since $P_1^*, P_2$ are alternating paths,
$P_1^*$ [$P_2$] has to take a right [left] turn in $\hat{F}$. Thus,
$P_1^*$ and $P_2$ diverge into different faces behind $\hat{F}$
without crossing in $\hat{F} \cup E(\hat{F})$.

If the slot conflict on $E(F)$ had been avoided, there would be no
crossing in $F$, and $P_1^*$ would run on the left side of $P_2^*$
(see Figure~\ref{fig:crossings}b). Then $P_1^*$ and $P_2$ would cross
in $\hat{F}$. Since the intersection pattern before $F$ and behind
$\hat{F}$ is not affected by the ambiguity in $F$, the total number of
crossings between $P_1$ and $P_2$ is not affected, either.
\end{proof}

Proposition~\ref{prop:egal} justifies the following definition.

\begin{definition}[Well-arranged graph]
We call a plane graph $G$ {\em well-arranged} if its collections of
EAPs are arrangements of alternating paths.
\label{def:well-arranged}
\end{definition}

\subsection{Partial cubes from well-arranged graphs}
\label{sub:ell_one}
Arrangements of alternating paths are generalizations of arrangements
of pseudolines~\cite{Bjoerner99a}. The latter are known to have duals
that are partial cubes~\cite{Eppstein2006a}. In this section we will
see that the dual of an arrangement of alternating paths is a partial
cube, too.

\begin{notation}
From now on $\mathcal{E}(G)$ denotes a collection of EAPs.
\label{def:multiset_embedded}
\end{notation}

The purpose of the following is to prepare the definition of
$\mathcal{E}(G)$'s dual (see Definition~\ref{def:DEAP}).

\begin{definition}[Domain $D(G)$ of $G$, facet of $\mathcal{E}(G)$, adjacent facets]
The domain $D(G)$ of $G$ is the set of points covered by the vertices,
edges and bounded faces of $G$. A facet of $\mathcal{E}(G)$ is a
(bounded) connected component (in $\RR^2$) of $D(G) \setminus
(\bigcup_{e \in E(G)} e \cup \bigcup_{v \in V(G)} v)$.  Two facets of
$\mathcal{E}(G)$ are adjacent if their boundaries share more than one
point.
\label{def:domain}  
\end{definition}

In the following, DEAP stands for Dual of Embedded Alternating Paths. 

\begin{definition}[DEAP graph $\subE{G}$ of $G$]
A DEAP graph $\subE{G}$ of $G$ is a plane graph that we obtain from
$G$ by placing one vertex into each facet of $\mathcal{E}(G)$ and
drawing edges between a pair $(u,v)$ of these vertices if the facets
containing $u$ and $v$ are adjacent in the sense of
Definition~\ref{def:domain}.  A vertex of $\subE{G}$ can also sit on
the boundary of a face as long as it does not sit on an EAP from
$\mathcal{E}(G)$ (for an example see the black vertex on the upper
left in Figure~\ref{fig:DEAP}a).
\label{def:DEAP}  
\end{definition}

Due to the intersection pattern of the EAPs in $G$'s bounded faces, as
specified in Section~\ref{sub:embed} and illustrated in
Figure~\ref{fig:intersection_patterns}, there are the following three
kinds of vertices in $V(\subE{G})$.

\begin{definition}[Primal, intermediate and star vertex of $\subE{G}$]~
\label{def:vertex_classes}
\begin{itemize}
\item \textbf{Primal vertices:} Vertices which represent a facet that
  contains a (unique) vertex $v$ of $G$ in its interior or on its
  boundary. As we do not care about the exact location of $\subE{G}$'s
  vertices, we may assume that the primal vertices of $\subE{G}$ are
  precisely the vertices of $G$.
\item \textbf{Intermediate vertices}: The neighbors of the primal
  vertices in $\subE{G}$.
\item \textbf{Star vertices} The remaining vertices in $\subE{G}$.
\end{itemize}
\end{definition}

For an example of a DEAP graph see Figure~\ref{fig:DEAP}, where the
black, gray and white vertices correspond to the primal, intermediate,
and star vertices, respectively.

\begin{figure}[]
\begin{center}
(a) \includegraphics[width=0.42\columnwidth]{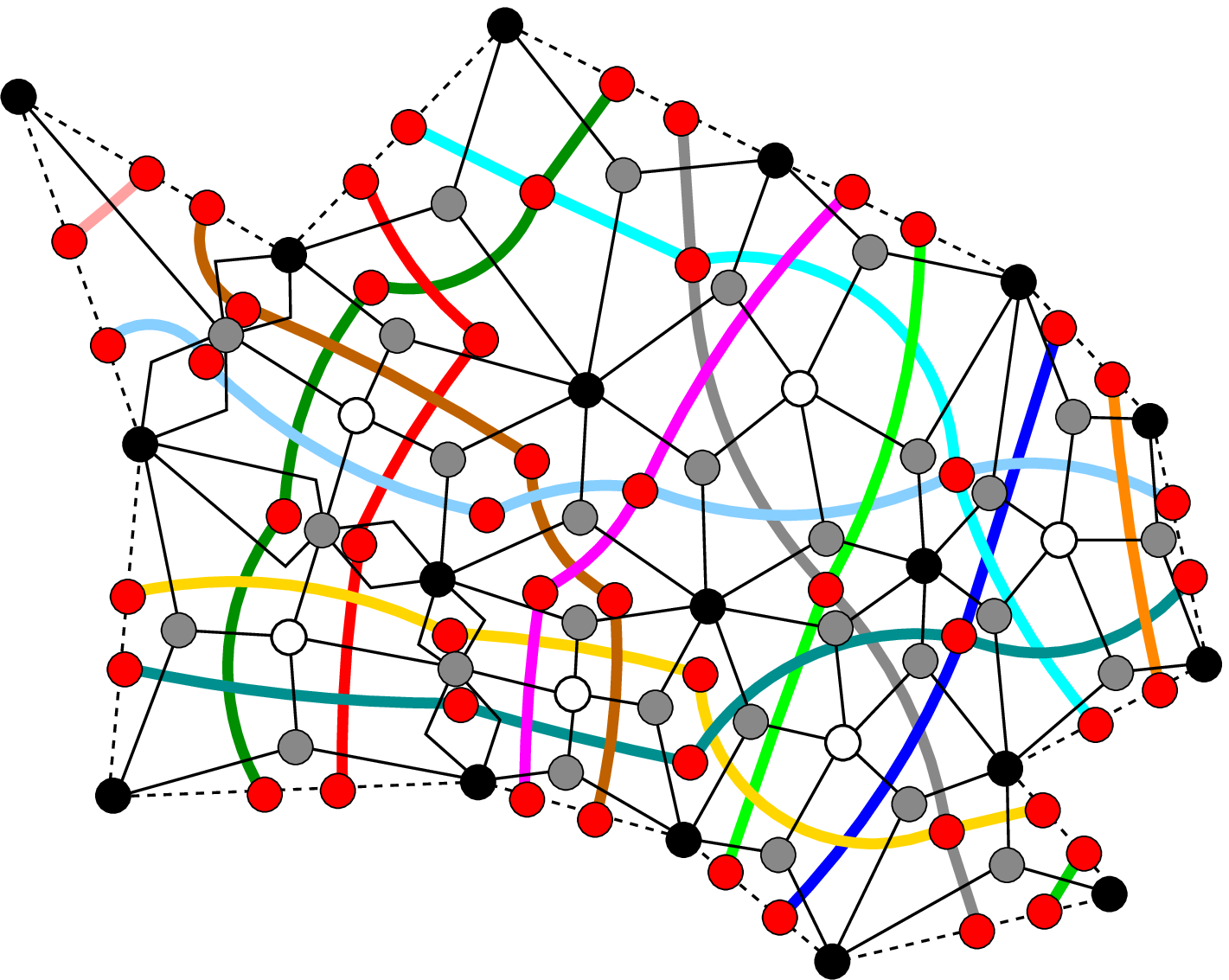}
(b) \includegraphics[width=0.42\columnwidth]{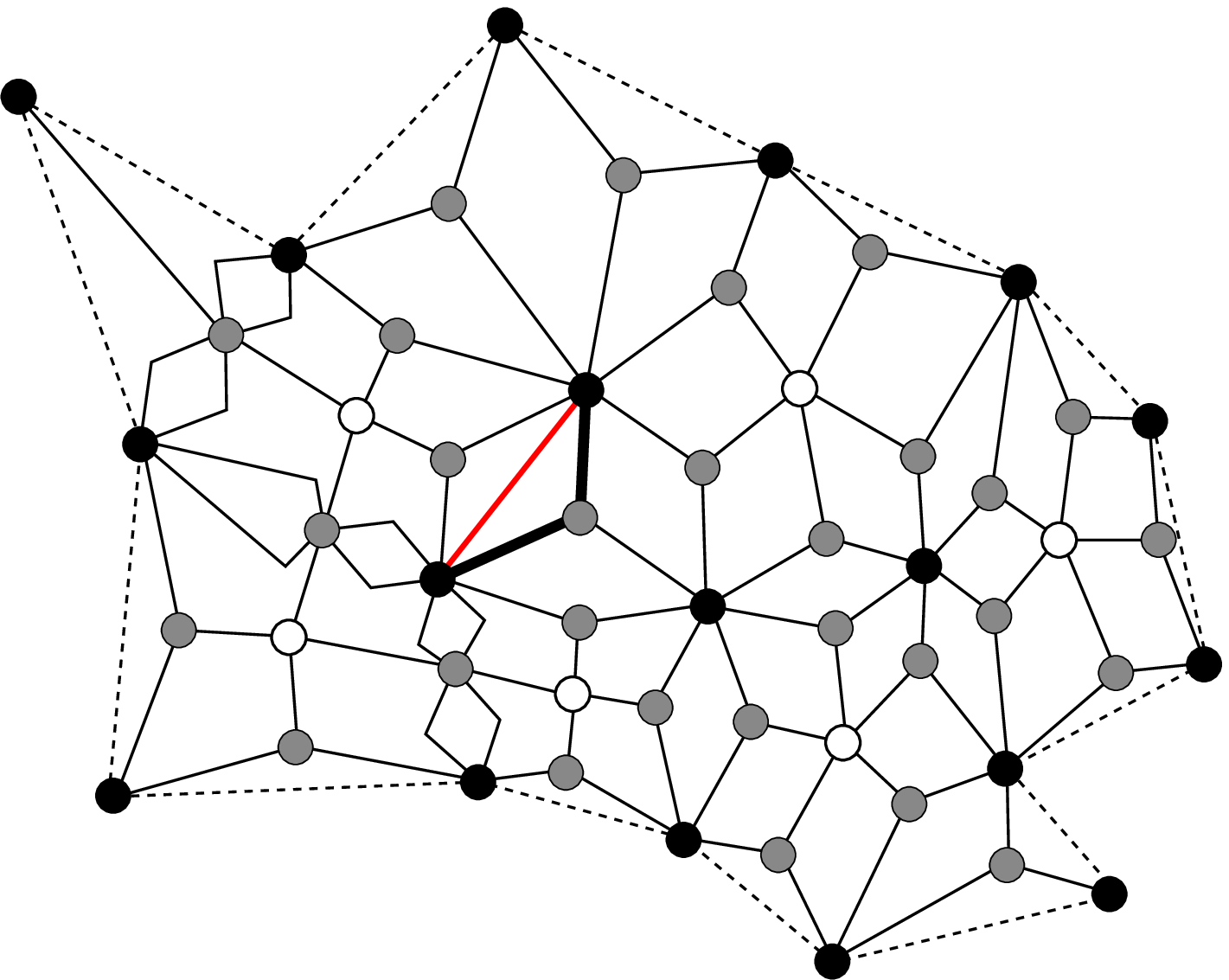}
\end{center}
\caption{DEAP graph $\subE{G}$ of the primal graph $G$ shown in
  Figure~\ref{fig:saws}a.  (a) Collection $\mathcal{E}(G)$ of EAPs:
  Red vertices, thick solid colored lines. DEAP graph $\subE{G}$:
  Black, gray and white vertices, thin black solid lines. The black,
  gray and white vertices are the primal, intermediate and star
  vertices, respectively. The dashed polygonal line delimits $D(G)$.
  (b) $\subE{G}$ only.  The red edge, however, is an edge of $G$. The
  path formed by the two bold black edges is an example of a path in
  $\subE{G}$ of length two that connects two primal vertices that are
  adjacent in $G$ via an intermediate vertex in $\subE{G}$.}
\label{fig:DEAP}
\end{figure}

\begin{definition}
Let $P$ be an EAP of $G$, and let $D_1$, $D_2$ denote the two
connected components $D(G) \setminus P$. The cut of $G = (V, E)$ given
by $P$ is $(V_1, V_2)$ with $V_i = \{v \in V \mid v \in D_i\}$, and
the cut of $\subE{G} = (\subE{V}, \subE{E})$ given by $P$ is $(V_1,
V_2)$ with $V_i = \{v \in \subE{V} \mid v \in D_i\}$.
\end{definition}

\begin{theorem}
\label{theo:arr2pc}
The DEAP graph $\subE{G}$ of a well-arranged plane graph $G$ is a
partial cube, the convex cuts of which are precisely the cuts given by
the EAPs of $G$.
\end{theorem}
\begin{proof}
We denote the Hamming distance by $h(\cdot,\cdot)$. To show that
$\subE{G}=(\subE{V},\subE{E})$ is a partial cube, it suffices to
specify a labeling $l : \subE{V} \mapsto \{0,1\}^{n}$ for some $n \in
\NN$ such that $d_{\subE{G}}(u,v) = h(l(u),l(v))$ for all $u, v \in
\subE{V}$ (see Section~\ref{sec:prelim}).

Let $\mathcal{E}(G) = \{P_1, \dots , P_n\}$. The length of the binary
vectors will be $n$. The entry of $l(v)$ indicates $v$'s position with
respect to the paths in $\mathcal{E}(G)$. Specifically, we arbitrarily
select one component of $D(G) \setminus P_i$ and set the $i$th entry
of $l(v)$ to one if the face represented by $v$ is part of the
selected component (zero otherwise).

It remains to show that $d_{\subE{G}}(u,v) = h(l(u),l(v))$ for any
pair $u\neq v \in V$. Since on any path of length $k$ from $u$ to $v$
in $\subE{G}$ it holds that $h(l(u),l(v)) \leq k$, we have
$d_{\subE{G}}(u,v) \geq h(l(u),l(v))$.

We assume $u \not= v$ (the case $u=v$ is trivial). To see that
$d_{\subE{G}}(u,v) = h(l(u),l(v))$, it suffices to show that $u$ has a
neighbor $u'$ such that $h(l(u'),l(v)) < h(l(u),l(v))$ (because then
there also exists $u''$ such that $h(l(u''),l(v)) < h(l(u'),l(v))$ and
so on until $v$ is reached in exactly $h(l(u),l(v))$ steps).

Let $F_u$ denote the facet of $\mathcal{E}(G)$ that is represented by
$u$, and $I(u)$ denote the set of indices of EAPs in $\mathcal{E}(G)$
that bound $F_u$.

\begin{enumerate}
\item If $u$ has only one neighbor $u'$, then $I(u)=\{k\}$ for some
  $k$, and the only vertex in one of the components of $D(G) \setminus
  P_k$ is $u$. For an example of $u'$ see the black vertex in the
  upper left corner of Figure~\ref{fig:DEAP}b. Since $l(u)$ and
  $l(u')$ differ only at position $k$, it must hold that
  $h(l(u'),l(v)) < h(l(u),l(v))$.
\item If $u$ has at least two neighbors, we first assume that none of
  the $P_k$ with $k \in I(u)$ cross each other (see
  Figure~\ref{fig:cases}a). Then $u$ is uniquely determined by the
  entries of $l(u)$ at the positions given by $I(u)$. Indeed, $F_u$ is
  then bounded by non-intersecting paths EAPs that run from a point on
  the border of $D(G)$ to another point on the border of $D(G)$, and
  only a vertex inside $F_u$ can have the same entries in $l(\cdot)$
  as $l(u)$ at the positions given by $I(u)$. Thus, since $u$ is the
  only vertex in $F_u$ and since $u \not= v$, $l(u)$ and $l(v)$ must
  differ at a position $k^*$ in $I(u)$ and, from the perspective of
  $u$, we find $u'$ in the face on the other side of $P_{k^*}$.
\item The remaining case is that $u$ has at least two neighbors and
  there exists at least one pair $(i,j) \in I(u) \times I(u)$, $i
  \not=j$, such that $P_i$ crosses $P_j$. Let $C$ denote the set of
  all such pairs. For any pair $(i,j) \in C$ the path $P_i$ crosses
  the path $P_j$ exactly once, because $\mathcal{E}(G)$ is an
  arrangement of alternating paths. Thus $P_i$ and $P_j$ subdivide
  $D(G)$ into four regions (see Figure~\ref{fig:cases}b), each of
  which is characterized by one of the four 0/1 combinations of vertex
  label entries at $i$ and at $j$. We may assume that $v$ is contained
  in the same region as $u$ for each pair $(i,j) \in C$ (otherwise we
  choose $u'$ on the other side of $P_i$ or $P_j$ and are done). Let
  $R$ be the intersection of all these regions, one region per pair in
  $C$.

If all $i \in I(u)$ are contained in at least one pair of $C$, we are
done. Indeed this means that $R=F_u$ and thus that $u$ is uniquely
determined by the entries of $l(u)$ at the positions given by
$I(u)$. We can then proceed as above. The remaining case is that there
exist $k \in I(u)$ such that $P_k$ does not intersect any $P_j$ with
$j \in I(u), j \not= k$ (see Figure~\ref{fig:cases}c). Recall that we
assumed $u \not= v$, $u, v \in R$, \ie $u$ and $v$ are separated by
$P_k$. Hence, the entries of $l(u)$ and $l(v)$ differ at a position $k
\in I(u)$, and $u'$ with $h(l(u'),l(v)) < h(l(u),l(v))$ can be reached
from $u$ by crossing $P_k$.
\end{enumerate}

So far we have shown that $\subE{G}$ is a partial cube and that the
cut of $\subE{G}$ given by any $P_i$ is precisely the cut that
determines the $i$-th entry of $l(u)$ for all $u \in \subE{V}$. Thus,
the cuts of $\subE{G}$ given by the $P_i$, $1 \leq i \leq n$, are
precisely the $n$ convex cuts of $\subE{G}$.
\end{proof}

\begin{figure}[h]
\begin{center}
(a) \includegraphics[height=0.2\columnwidth]{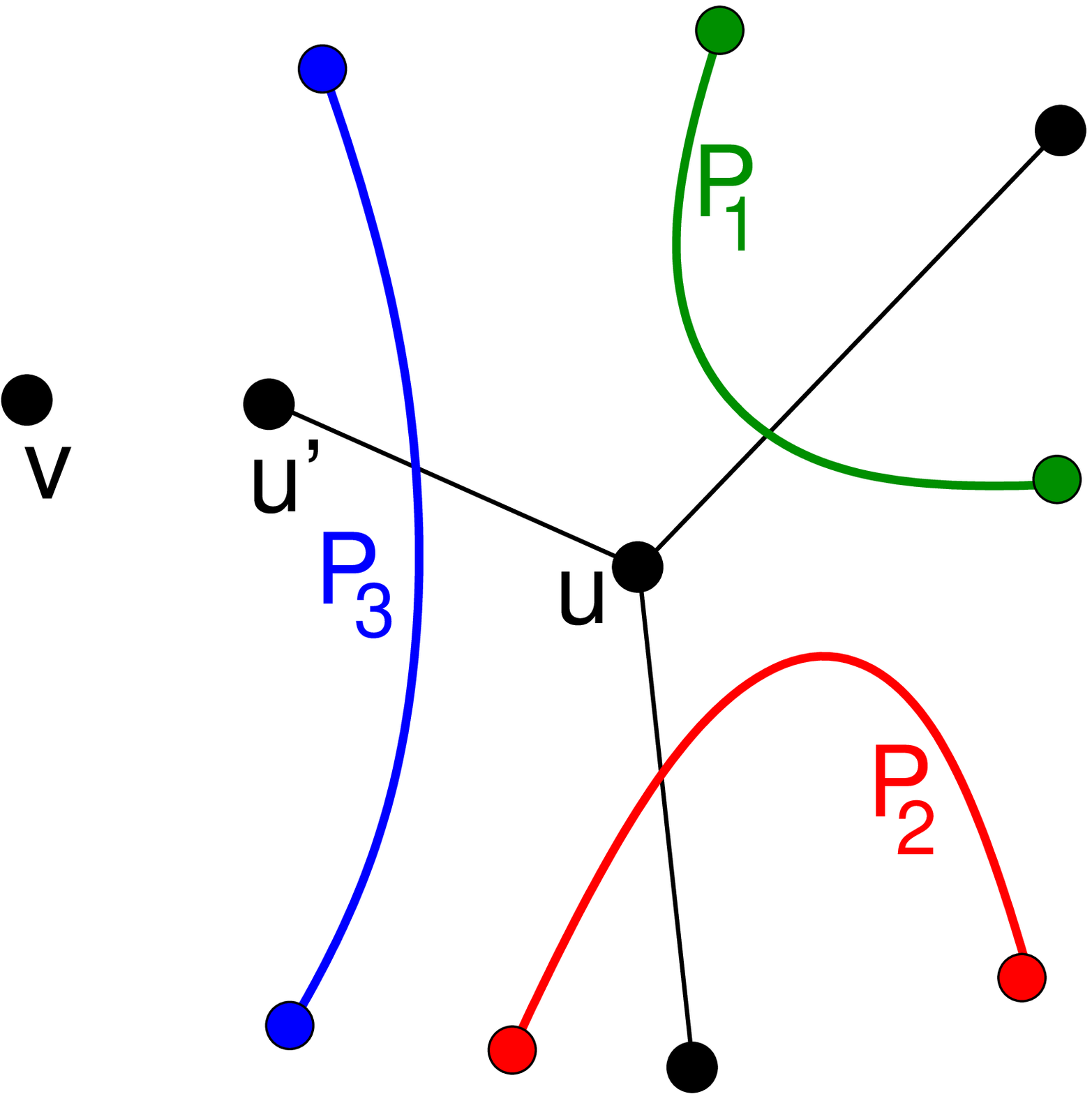}\qquad
(b) \includegraphics[height=0.2\columnwidth]{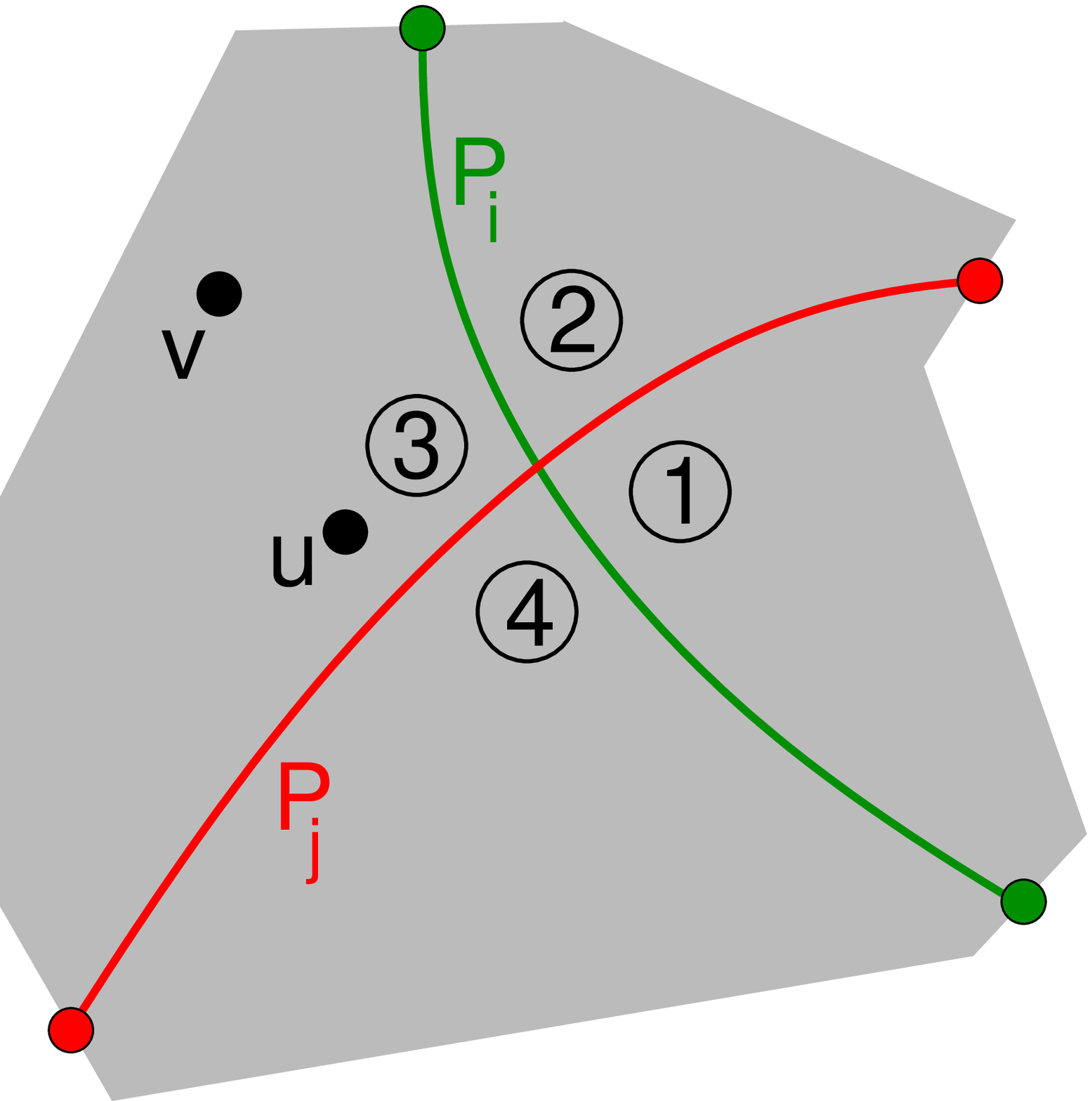}\qquad
(c) \includegraphics[height=0.2\columnwidth]{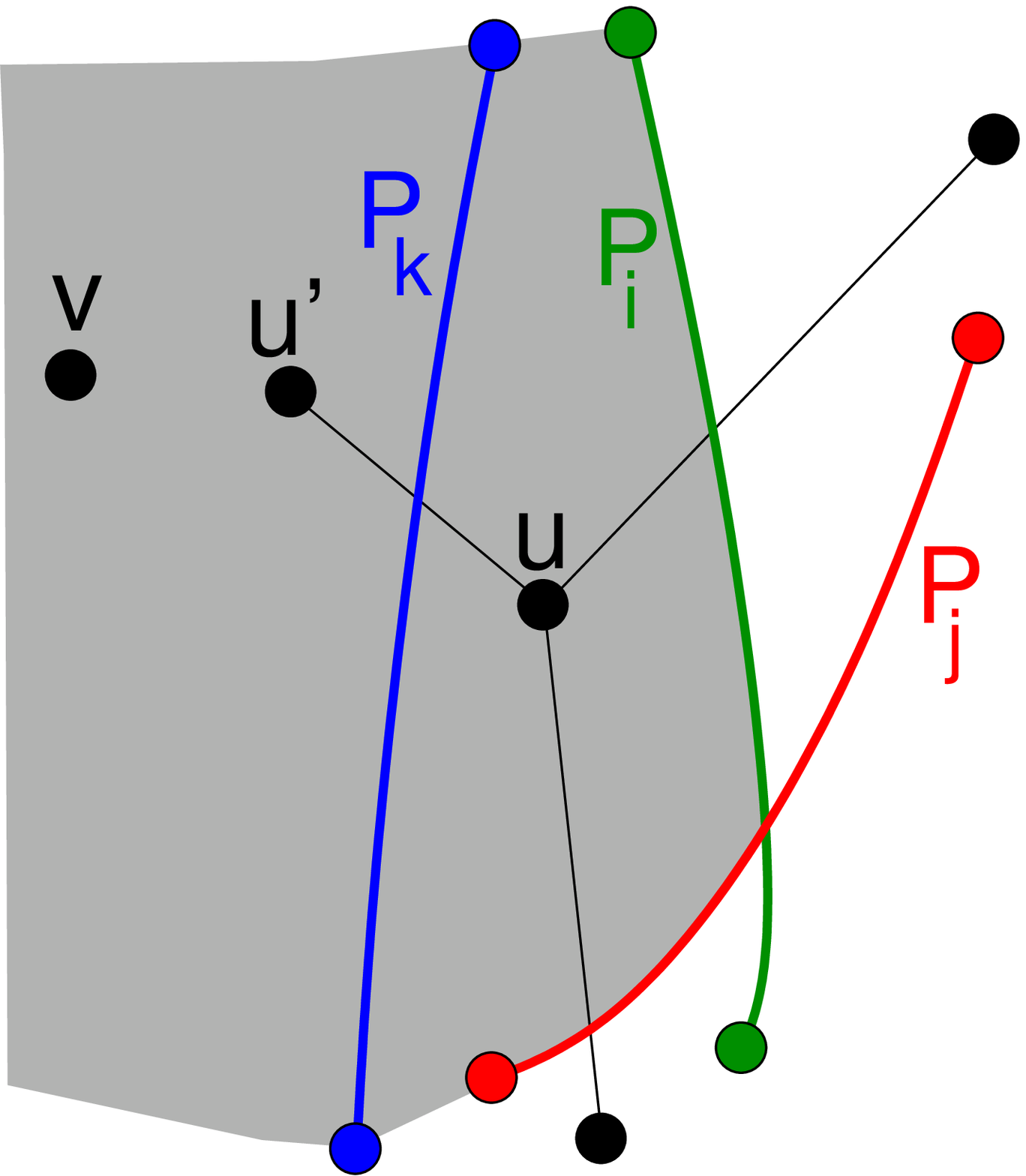}
\caption{\label{fig:cases} Subgraph of $\subE{G}$: black vertices and
  black edges. Subset of EAPs: colored vertices and edges. (a)
  Illustration to item 2 in proof of Theorem~\ref{theo:arr2pc}. (b,c)
  Illustrations to item 3 in proof of Theorem~\ref{theo:arr2pc}. The
  shaded regions in (b) and (c) indicate $D(G)$ and $R$,
  respectively.}
\end{center}
\end{figure}


Recall that the cut-sets of the convex cuts of a partial cube are
precisely the equivalence classes of the \djoko (see
Section~\ref{sec:prelim}). Thus, Theorem~\ref{theo:arr2pc} yields the
following.

\begin{corollary}
\label{coro:pc2}
If $G$ is well-arranged, there exists a one-to-one correspondence
between the EAPs of $G$ and the equivalence classes of the \djoko on
$\subE{G}$. Specifically, an equivalence class of the \djoko on
$\subE{G}$ is given by the set of edges intersected by an EAP of $G$,
and vice versa.
\end{corollary}

\subsection{Convex Cuts from Alternating Paths}
\label{sub:shortcut}
In case of $G$ being well-arranged, the subgraph defined below will
serve as a stepping stone for relating the convex cuts of $\subE{G}$
to those of $G$ (see Lemma~\ref{lemma:iso_sub}).

\begin{definition}[Subgraph $\widetilde{\subE{G}}$ of $\subE{G}$]
$\widetilde{\subE{G}}$ is the graph obtained from $\subE{G}$ by
  deleting all star vertices and replacing parallel edges by single
  edges.
\label{def:iso_sub}  
\end{definition}


In the following recall that we identified the primary vertices in
$\subE{G}$ with the vertices of $G$.

\begin{lemma}
\label{lemma:iso_sub}
If $G$ is well-arranged, then $d_G(u, v) = \frac{1}{2}
d_{\widetilde{\subE{G}}}(u, v) = \frac{1}{2} d_{\subE{G}}(u, v)$.
\end{lemma}

\begin{proof}~\vspace{-0.5cm}\\
\begin{itemize}
\item For any face $F$ of $G$, let $G^F$ [$\subE{G}^F$] denote the
  subgraph of $G$ [$\subE{G}$] that is contained in $\closu{F}$. From
  $G$ being well-arranged follows that $G^F$ is well arranged, and
  Theorem~\ref{theo:arr2pc} yields that $\subE{G}^F$ is a partial
  cube.
\item Let $u$, $v$ be vertices of $G^F$. Then, due to the EAPs being
  embedded alternating paths, there exists a path $P$ from $u$ to $v$
  in $G$ that (i) contains only edges from $E(F)$ and (ii) crosses any
  of the EAPs through $F$ at most once. Hence, there exists a path
  $\widetilde{\subE{P}}$ from $u$ to $v$ in $\widetilde{\subE{G}}$
  (see Definition~\ref{def:iso_sub}) that (i) is a subdivision of $P$
  with every other vertex being an intermediate vertex and (ii) also
  crosses any of the EAPs through $F$ at most once. Since $\subE{G}^F$
  is a partial cube, $\widetilde{\subE{P}}$ is a shortest path in
  $\subE{G}^F$. Its length is $2 d_G(u, v)$.
\item Let $u$, $v$ be vertices of $G$. Using the previous item
  repeatedly, we get that there exists a shortest path
  $\widetilde{\subE{P}}^*$ from $u$ to $v$ in $\widetilde{\subE{G}}$
  whose length is $d_{\subE{G}}(u, v) = d_{\widetilde{\subE{G}}}(u,
  v)$. On $\widetilde{\subE{P}}^*$ the vertices in $G$ alternate with
  intermediate vertices, and for any edge $\{x, y\}$ in $G$ there
  exists a path of length two between $x$ and $y$ with the vertex in
  the middle being an intermediate vertex. Thus, $d_G(u, v) =
  \frac{1}{2} d_{\widetilde{\subE{G}}}(u, v) = \frac{1}{2}
  d_{\subE{G}}(u, v)$.
\end{itemize}
\end{proof}

\begin{proposition}
\label{prop:well-arr}
If $G$ is well-arranged, any cut given by an EAP of $G$ is convex.
\end{proposition}

\begin{proof}
Let $(V_1, V_2)$ be the cut of $G$ given by an EAP $P$ of $G = \{V,
E\}$. Without loss of generality let $u_1, v_1 \in V_1$. We have to
show that any shortest path from $u_1$ to $v_1$ contains only vertices
from $V_1$.

We assume the opposite, \ie that there exists a shortest path $S$ from
$u_1$ to $v_1$ in $G$ that contains a vertex from $V_2$, \ie a vertex
on the other side of the EAP $P$. This shortest path can be turned
into a path $\subE{S}$ in $\subE{G}$ with twice the length by
inserting a vertex from $v' \setminus V$ between any pair of
consecutive vertices on $S$. Lemma~\ref{lemma:iso_sub} then yields
that $\subE{S}$ is a shortest path in $\subE{G}$ that crosses the EAP
$P$ twice. This is a contradiction to Theorem~\ref{theo:arr2pc}.
\end{proof}

%
%
%
\section{Convex cuts of bipartite graphs}
\label{sec:bipartite}
Let $H'=(V,E)$ be a bipartite but not necessarily plane graph. As
mentioned in Section~\ref{sec:prelim}, any edge $e=\{a,b\}$ of $H'$
gives rise to a cut of $H'$ into $W_{ab}$ and $W_{ba}$.  The cut-set
of this cut is $C_e = \{f \in E \mid e~\theta'~f\}$, where $\theta'$
is the \djoko on $H'$. Note that $C_e = \{f=\{u,v\} \in E \mid
d_{H'}(a,u)=d_{H'}(b,v)\}$. In the following we characterize the
cut-sets of the {\em convex} cuts of $H'$. This characterization is
key to finding all convex cuts of a bipartite graph in $\bigO(\vert E
\vert^3)$ time.

\begin{lemma}
\label{lemma-pairwise-djoko}
Let $H'=(V,E)$ be a bipartite graph, and let $e \in E$. Then $C_e$ is
the cut-set of a convex cut of $H'$ if and only if $f~\theta'~\hat{f}$
for all $f, \hat{f} \in C_e$.
\end{lemma}

\begin{proof}
''$\Leftarrow$`` Let $e=\{a,b\}$, and assume that the cut with
cut-set $C_e$ is not convex. Then there exists a shortest path
$P=\{v_1, \dots v_n\}$ with both end vertices in, say, $W_{ab}$ such
that $P$ has a vertex in $W_{ba}$. Let $i$ be the smallest index such
that $v_i \in W_{ba}$, and let $j$ be the smallest index greater than
$i$ such that $v_j \in W_{ab}$. Then $f=\{v_{i-1}, v_i\}$ and
$\hat{f}=\{v_{j-1}, v_j\}$ are contained in $C_e$. We use now a result
by Ovchinnikov~\cite[Lemma~3.5]{Ovchinnikov2008a}, which states that
no pair of edges on a shortest path are related by $\theta'$, \ie
$f~\theta'~\hat{f}$ does {\em not} hold.

''$\Rightarrow$`` Let $f=\{u,v\}$ and $\hat{f}=\{\hat{u},\hat{v}\}$ be
edges in $C_e$ such that $f~\theta'~\hat{f}$ does not hold. Without
loss of generality assume $u, \hat{u} \in W_{ab}$, $v, \hat{v} \in
W_{ba}$ and $d_{H'}(u,\hat{u}) < d_{H'}(v,\hat{v})$. Due to $H'$ being
bipartite, both distances are even or both are odd. Hence,
$d_{H'}(v,\hat{v}) - d_{H'}(u,\hat{u}) \geq 2$. Consider the path
$\hat{P}$ from $v$ via $f$ to $u$, then along a shortest path from $u$
to $\hat{u}$ and finally from $\hat{u}$ to $\hat{v}$ via
$\hat{f}$. This path has length $d_{H'}(u,\hat{u}) + 2 \leq
d_{H'}(v,\hat{v})$. Thus, $\hat{P}$ is a shortest path from $v$ to
$\hat{v}$ (and $d_{H'}(v,\hat{v}) - d_{H'}(u,\hat{u}) = 2$). The path
$\hat{P}$ is a shortest path from $v \in W_{ba}$ via $u, \hat{u} \in
W_{ab}$ to $\hat{v} \in W_{ba}$, so that $C_e$ is not the cut-set of a
convex cut.
\end{proof}

Lemma~\ref{lemma-pairwise-djoko} suggests to determine the convex cuts
of $H'$ as sketched in Algorithm~\ref{algo:bipartite} by checking for each cut-set $C_{e^i}$
if the cut-sets of the contained edges $f^j$ all coincide.

\begin{algorithm}
\caption{Find all cut-sets of convex cuts of a bipartite graph $H'$}
\label{algo:bipartite}
\begin{algorithmic}[1]
\Procedure{EvaluateCutSets}{bipartite graph $H'$} \newline
\Comment{Computes $C_{e^i}$ for each edge $e^i$ and stores
  in $isConvex[i]$ if $C_{e^i}$ is the cut-set of a convex cut}
\State Let $e^1, \dots\, e^m$ denote the edges of $H'$; initialize all $m$
entries of the array $isConvex$ as $true$
\For{$i=1, \dots, m$}
 \State Determine $C_{e^i}=\{f^j~|~e^i~\theta~f^j\}$
 \ForAll{$f^j \neq e^i$}
  \State Determine $C_{f^j} = \{g^k~|~f^j~\theta~g^k \}$
  \If{$C_{f^j} \neq C_{e^i}$}
   \State $isConvex[i] := false$
   \State \textbf{break}
  \EndIf
 \EndFor 
\EndFor
\EndProcedure
\end{algorithmic}
\end{algorithm}

\begin{theorem}
\label{thm:algo-time}
All convex cuts of a bipartite graph can be found using $\bigO(|E|^3)$
time and $\bigO(|E|)$ space.
\end{theorem}
\begin{proof}
We use Algorithm~\ref{algo:bipartite}. The correctness follows from
Lemma~\ref{lemma-pairwise-djoko} and the symmetry of the
\djoko. Regarding the running time, observe that for any edge $e^i$,
the (not necessarily convex) cut-set $C_{e^i}$ can be determined using
breadth-first search to compute the distances of any vertex to the end
vertices of $e^i$. Thus, any $C_{e^i}$ can be determined in time
$\bigO(\vert E \vert)$. We mark [unmark] the edges in $C_{e^i}$ when
entering [exiting] the inner for-loop, which has $\bigO(|E|))$
iterations.

The time complexity $\bigO(|E|))$ of an iteration of the inner loop is
due to the calculation of $C_{f^j}$. The test whether $C_{f^j} =
C_{e^i}$ is done on the fly: any time a new edge of $C_{f^j}$ is
found, we only check if the edge has a mark. We have $C_{f^j} =
C_{e^i}$ if and only if all new edges are marked. This follows from
the fact that no proper subset of $C_{e^i}$ can be a cut-set (the
subgraphs induced by $W_{ab}$ and $W_{ba}$ are connected). Hence,
Algorithm~\ref{algo:bipartite} runs in $\bigO(\vert E \vert^3)$ time.
Since no more than two cut-sets (with $\bigO(|E|)$ edges each) have to
be stored at the same time, the space complexity of
Algorithm~\ref{algo:bipartite} is $\bigO(|E|)$.
\end{proof}
A simple loop-parallelization over the edges in line $3$ leads to a
parallel running time of $\bigO(|E|^2)$ with $\bigO(|E|)$
processors. If one is willing to spend more processors and a quadratic
amount of memory, then even faster parallelizations are
possible. Since they use standard PRAM results, we forgo their
description.

\section{Convex cuts of general graphs}
\label{sec:general}
In Theorem~\ref{theorem_subdiv} of this section we characterize the
cut-sets of the convex cuts of a general graph $H$ in terms of two
binary relations on edges: the \djoko and the relation $\tau$ (for
$\tau$ see Definition~\ref{def:tau}). We will use
Theorem~\ref{theorem_subdiv} in Section~\ref{sec:general_plane} to
find all convex cuts of a {\em plane} graph.

While the relation $\tau$ is applied to the edges of $H$, the \djoko is
applied to the edges of a bipartite subdivision $H'$ of
$H$. Specifically, $H'$ is the graph that one obtains from $H$ by
subdividing each edge of $H$ into two edges. An edge $e$ in $H$ that
is subdivided into edges $e_1, e_2$ of $H'$ is called {\em parent} of
$e_1, e_2$, and $e_1, e_2$ are called {\em children} of $e$.

\begin{definition}[Relation $\tau$]
\label{def:tau}
Let $e=\{u_e, v_e\}$ and $f=\{u_f, v_f\}$ be edges of $H$. Then,
$e~\tau~f$ iff $d_H(u_e, u_f) = d_H(v_e, v_f) = d_H(u_e, v_f) =
d_H(v_e,u_f)$.
\end{definition}

The next lemma follows directly from the definition of $\theta'$ and
$\tau$.

\begin{lemma}
\label{lemma:theta_tau}
If $e~\tau~f$, then none of the children of $e$ is
$\theta'$-related to a child of $f$.
\end{lemma}

\begin{theorem}
\label{theorem_subdiv}
A cut of $H$ with cut-set $C$ is convex if and only if for all $e, f
\in C$ it holds either $e~\tau~f$ or that there exists a child $e'$ of
$e$ and a child $f'$ of $f$ such that $e'~\theta'~f'$.
\end{theorem}

To simplify the proof of Theorem~\ref{theorem_subdiv}, we first
establish the following result.
\begin{lemma}
\label{lem:aux-equiv}
Let $e = \{u_e, v_e\}$ and $f = \{u_f, v_f\}$ be edges of $H$. Then
the following is equivalent.
\begin{itemize}
\item[(i)] There exists a child $\{a',b'\}$ of $e$ with $a'$ closer to
  $u_e$ than $b'$ and a child $\{c',d'\}$ of $f=\{u_f, v_f\}$ with
  $c'$ closer to $u_f$ than $d'$ such that $d_{H'}(a', c') =
  d_{H'}(b', d')$.
\item[(ii)] $e~\tau~f$ or there exists a child $e'$ of $e$ and a child
  $f'$ of $f$ with $e'~\theta'~f'$.
\end{itemize}
\end{lemma}
\begin{proof}
We denote by $w_e'$ [$w_f'$] the vertex of $H'$ that subdivides $e$
[$f$]. Without loss of generality we assume that $d_H(u_e, u_f) \geq
d_H(v_e, v_f)$ (see Figures~\ref{fig:notation-children}a,b).

To prove ''$(i) \Rightarrow (ii)$``, let $d_{H'}(a', c') = d_{H'}(b',
d')$.
\begin{itemize}
\item We first assume $d_H(u_e, v_f) \neq d_H(u_e, u_f)$ and $d_H(v_e,
  u_f) \neq d_H(u_e, u_f)$ (see
  Figure~\ref{fig:notation-children}a). Then our assumption $d_H(u_e,
  v_f) \geq d_H(u_e, u_f)$, in conjunction with the fact that $w_e'$
  and $w_f'$ both have degree two, imply that there exists a shortest
  path from $w_e'$ via $v_e$ and $v_f$ to $w_f'$. The equality
  $d_{H'}(a', c') = d_{H'}(b', d')$ yields $d_H(u_e, u_f) = d_H(v_e,
  v_f)$. Indeed, $d_H(u_e, u_f) > d_H(v_e, v_f)$ would mean that there
  exists a shortest path from $u_e$ via $v_e$ to $u_f$, a
  contradiction to $d_{H'}(a', c') = d_{H'}(b', d')$ (recall that $a'$
  is closer to $u_e$ than $b'$ and that $c'$ is closer to $u_f$ than
  $d'$).

  From $d_H(u_e, v_f) \neq d_H(u_e, u_f)$, $d_H(v_e, u_f) \neq
  d_H(u_e, u_f)$, and $d_{H'}(a', c') = d_{H'}(b', d')$ we conclude
  $a' = u_e$, $b' = w_e'$, $c' = w_f'$ and $d' = v_f$. In particular,
  $\{a', b'\}~\theta~\{c', d'\}$.
\item If $d_H(u_e, v_f) \neq d_H(u_e, u_f)$ and $d_H(v_e, u_f) =
  d_H(u_e, u_f)$ holds (see Figure~\ref{fig:notation-children}b),
  there must exist a shortest path from $v_e$ to $u_f$ via $v_f$ and
  $w_f'$ because otherwise the prerequisite $d_{H'}(a', c') =
  d_{H'}(b', d')$ would not hold. Then $d_H(v_e, v_f) = d_H(v_e, u_f)
  - 1 = d_H(u_e, u_f) -1$. Thus $d_{H'}(a', c') = d_{H'}(b', d')$
  implies $a'=u_e$, $b'=w_e'$, $c'=v_e$ and $d'=w_f'$. In particular,
  $\{a', b'\}~\theta~\{c', d'\}$.

\item If $d_H(u_e, v_f) \neq d_H(u_e, u_f)$ and $d_H(v_e, u_f) =
  d_H(u_e, u_f)$, then we can proceed as in the previous item and conclude
  that we cannot have $d_{H'}(a', c') = d_{H'}(b', d')$ or that
  $d_H(u_e, u_f) = d_H(v_e, v_f) -1$, a contradiction to our
  assumption $d_H(u_e, u_f) \geq d_H(v_e, v_f)$.

\item If $d_H(u_e, v_f) = d_H(u_e, u_f)$ and $d_H(v_e, u_f) = d_H(u_e,
  u_f)$, then $e~\tau~f$. Indeed, due to the definition of $\tau$, it
  suffices to show that $d_H(u_e, u_f) = d_H(v_e, v_f)$. We assume the
  opposite, \ie that $d_H(u_e, u_f) > d_H(v_e, v_f)$ (recall that we
  assume $d_H(u_e, u_f) \geq d_H(v_e, v_f)$). Thus, $d_H(u_e, u_f) >
  d_H(v_e, v_f) \geq d_H(v_e, u_f) - 1$, \ie $d_H(u_e, u_f) \geq
  d_H(v_e, u_f)$. If $d_H(u_e, u_f) = d_H(v_e, u_f)$, and thus
  $d_H(v_e, v_f) < d_H(u_e, u_f) = d_H(v_e, u_f) = d_H(u_e, v_f)$,
  there exists a shortest path from $v_e$ via $v_f$ to $u_f$, a
  contradiction to $d_{H'}(a', c') = d_{H'}(b', d')$. Otherwise,
  $d_H(u_e, u_f) > d_H(v_e, u_f) = d_H(u_e, u_f)$. If $d_H(v_e, v_f) <
  d_H(v_e, u_f)$, there exists a shortest path from $v_e$ via $v_f$ to
  $u_f$ (recall that $d_H(u_e, u_f) > d_H(v_e, v_f)$), a contradiction
  to $d_{H'}(a', c') = d_{H'}(b', d')$. The case $d_H(v_e, v_f) \geq
  d_H(v_e, u_f)$ does not occur since, due to $d_H(v_e, u_f) =
  d_H(u_e, v_f)$, we would get $d_H(v_e, v_f) \geq d_H(u_e, u_f)$.
\end{itemize}

To prove ''$(ii) \Rightarrow (i)$``, assume that $e'~\theta~f'$ for a
child $e'$ of $e$ and a child $f'$ of $f$. Then the end vertices of
$a'$, $b'$ of $e'$ and $c'$, $d'$ of $f'$ fulfill $d_{H'}(a', c') =
d_{H'}(b', d')$ by definition of $\theta'$. If $e~\tau~f$, then we set
$a' = u_e$, $b' = w_e'$, $c' = w_f'$, and $d' = v_f$ as in
Figure~\ref{fig:notation-children}a.

The ''either`` in the claim follows from
Lemma~\ref{lemma:theta_tau}.
\end{proof}

\begin{figure}[bt]
\begin{center}
\includegraphics[height=0.25\columnwidth]{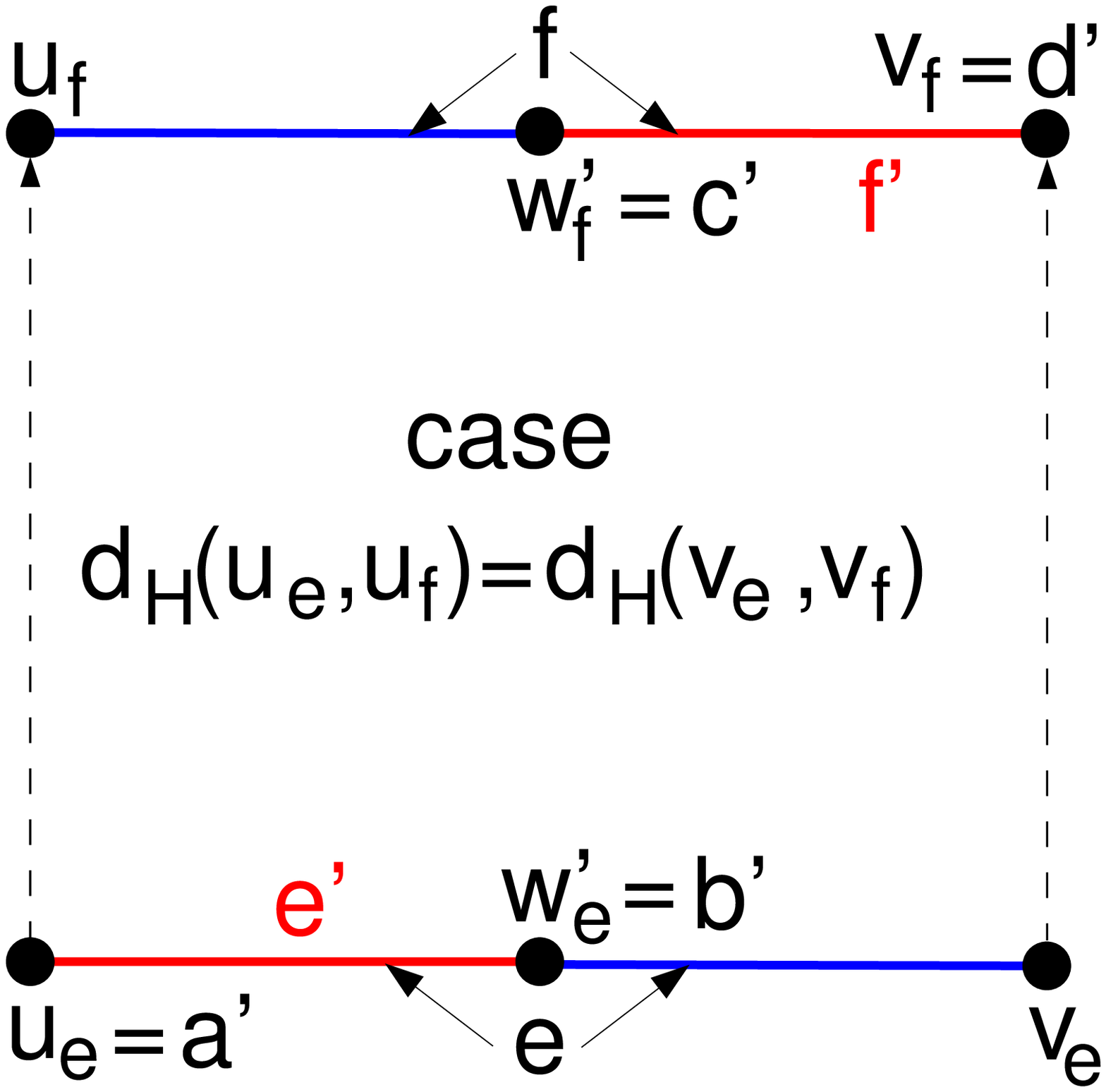}\qquad
\includegraphics[height=0.25\columnwidth]{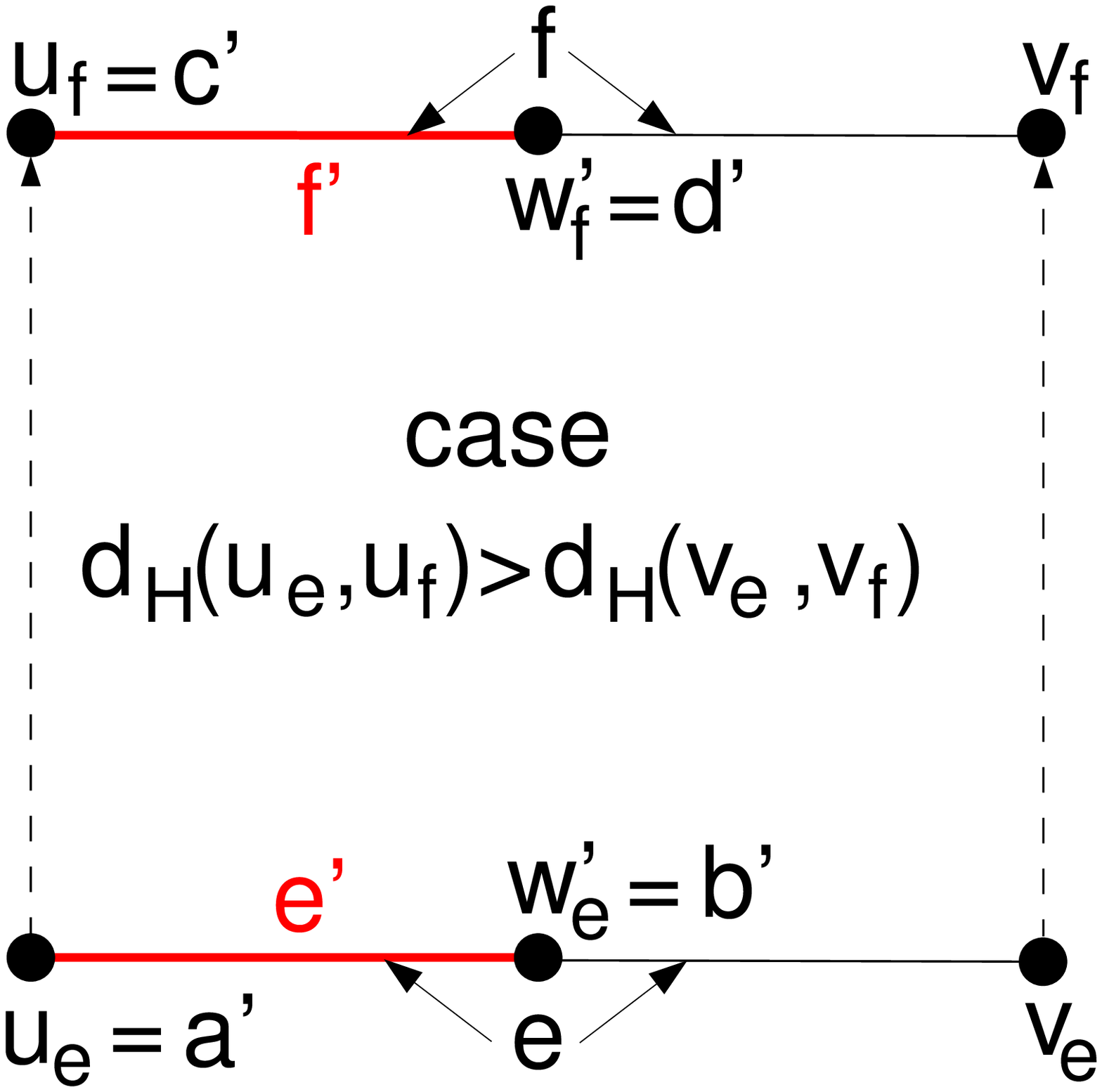}
\caption{\label{fig:notation-children} Illustrations to proof of
  Theorem~\ref{theorem_subdiv}. (a) If $e~\tau~f$ does not hold, at
  least one child of $e$ is $\theta'$-related to a child of $f$, \eg
  $e'~\theta'~f'$. (b) $e~\tau~f$ does not hold, and exactly one child
  of $e$ is related to a child of $f$ (here $e'~\theta'~f'$).}
\end{center}
\end{figure}

\noindent{\em Proof of Theorem~\ref{theorem_subdiv}}.\\
To prove necessity, let $C$ be the cut-set of a {\em convex} cut that
partitions $V$ into $V_1$ and $V_2$, and let $e,f \in C$ (see
Figure~\ref{fig:notation-children}). Thanks to
Lemma~\ref{lem:aux-equiv} it suffices to find a child $\{a',b'\}$ of
$e$ and a child $\{c',d'\}$ of $f$ such that $d_{H'}(a', c') =
d_{H'}(b', d')$.  Let $w_e'$ [$w_f'$] denote the vertex of $H'$ that
subdivides $e$ [$f$]. Without loss of generality we assume $u_e, u_f
\in V_1$ and $v_e, v_f \in V_2$. Since $C$ is the cut-set of a convex
cut, we know that $d_H(u_e, u_f)$ and $d_H(v_e, v_f)$ differ by at
most one.
\begin{enumerate}
\item If $d_H(v_e, v_f) = d_H(u_e,
u_f)$, let $\{a',b'\} = \{u_e, w_e'\}$ and $\{c',d'\} = \{w_f', v_f\}$
(this is the case illustrated in Figure~\ref{fig:notation-children}a).
Then, due to the degrees of $w_e'$ and $w_f'$ being two, $d_{H'}(a',
c') = d_{H'}(u_e, w_f')=d_{H'}(w_e', v_f) = d_{H'}(b', d')$.
\item If $d_H(u_e, u_f)$ and $d_H(v_e, v_f)$ differ by exactly one, we may
assume without loss of generality that $d_H(v_e, v_f) = d_H(u_e, u_f)
+1$.  Set $\{a',b'\} = \{w_e', v_e\}$ and $\{c',d'\} = \{w_f',
v_f\}$. Then, due to the degrees of $w_e'$ and $w_f'$ being two,
$d_{H'}(a', c') =d_{H'}(w_e', w_f') = d_{H'}(v_e, v_f) = d_{H'}(b',
d')$.
\end{enumerate}


Conversely, to prove sufficiency, let $C$ be the cut-set of a cut that
partitions $V$ into $V_1$ and $V_2$. We distinguish the two cases
of the prerequisite.

\begin{itemize}
\item Case 1 ($e~\tau~f$): Then we have $d_{H}(u_e, u_f) = d_{H}(v_e, v_f)$ by
  definition of $\tau$.
\item Case 2 (there exists a child $e'$ of $e$
  and a child $f'$ of $f$ such that $e'~\theta'~f'$): As above we
  assume without loss of generality that $u_e, u_f \in V_1$ and $v_e,
  v_f \in V_2$.  There are four possibilities for the positions of
  $e'$ and $f'$ within $e$ and $f$, only two of which need to be
  considered due to symmetry.

\begin{enumerate}
\item $e'=\{u_e, w_e'\}$ and $f'=\{w_f', v_f\}$. Since the degrees of
  $w_e'$ and $w_f'$ are two, and since $e'~\theta'~f'$, any shortest
  path from $u_e$ to $w_f'$ runs via $u_f$, and any shortest path from
  $w_e'$ to $v_f$ runs via $v_e$. Hence, $d_{H'}(u_e, u_f)=d_{H'}(u_e,
  w_f')-1 = d_{H'}(w_e', v_f)-1 = d_{H'}(v_e, v_f)$.
\item $e'=\{u_e, w_e'\}$ and $f'=\{u_f, w_f'\}$. In this case
  $d_{H'}(u_e, u_f) = d_{H'}(w_e', w_f') = d_{H'}(v_e, v_f) \pm 2$.
\end{enumerate}
\end{itemize}

To summarize Case 1 and Case 2, we always have $d_{H'}(u_e, u_f) =
d_{H'}(v_e, v_f) \pm 2$.

Due to $d_{H'}(u, v) = 2 d_H(u, v)$ for all vertices $u, v$ of $H$, we
have that $d_{H}(u_e, u_f) = d_{H}(v_e, v_f) \pm 1$ for all $e=\{u_e,
v_e\}, f=\{u_f, v_f\}$ in the cut-set $C$. Hence, any shortest path
with end vertices in $V_1$ [$V_2$] stays within $V_1$ [$V_2$], \ie $C$
is the cut-set of a convex cut.

%

%
%
%
\section{Convex cuts of plane graphs}
\label{sec:general_plane}
In this section $G=(V,E)$ is a plane graph with the restrictions
formulated in Section~\ref{sec:prelim}. Recall that the restrictions
are not essential for finding convex cuts.


We search for cut-sets of convex cuts of $G$ by brachiating from an
edge $e_0$ of $G$ via a bounded face $F_0$ of $G$, \ie $e_0 \in
E(F_0)$, to an edge $e_1$ on $E(F_0) \cap E(F_1)$ for some bounded
face $F_1$ of $G$, and so on. Theorem~\ref{theorem_subdiv} in this
paper and Lemma 2 in~\cite{Chepoi97a} restrict and thus guide the
brachiating. The latter lemma says that for the cut-set $C$ of any
convex cut and any bounded face $F$ we have that $\vert C \cap E(F)
\vert$ equals zero or two. Our approach to finding (cut-sets of)
convex cuts through brachiating suggests the following notation.

\begin{notation}
\label{not:C}
$C$ denotes a cut-set of a cut of $G$ and is written as a non-cyclic
or cyclic (simple cycle) sequence $(e_0, \dots, e_{\vert C \vert -
  1})$. If $C$ is non-cyclic, there exist bounded faces $F_0, \dots,
F_{\vert C \vert-2}$ of $G$ such that $e_{i-1} \in E(F_{i-1}) \cap
E(F_i)$. If $C$ is cyclic there exist bounded faces $F_0, \dots,
F_{\vert C \vert-1}$ such that $e_{i-1} \in E(F_{i-1}) \cap E(F_i)$,
and indices are modulo $\vert C \vert$.
\end{notation}

In particular, $C \cap E(F_{\infty}) = \{e_0, e_{\vert C \vert-1}\}$
for non-cyclic $C$ and $C \cap E(F_{\infty}) = \emptyset$ for cyclic
$C$.

Analogous to Section~\ref{sec:general}, $G'=(V',E')$ denotes the
(plane bipartite) graph that one obtains from $G$ by placing a new
vertex into the interior of each edge of $G$.

\begin{definition}[$e_0^l$, $e_0^r$, $C_l'$, $C_r'$, $C_l$, $C_r$, $C_{\tau}$]
\label{def:equiv-cuts}
The left [right] child of $e_0$ when standing on $e_0$ and looking
into $F_0$ is denoted by $e_0^l$ [$e_0^r$]. Furthermore, $C_l'$
[$C_r'$] denotes the set of edges in $E'$ that are $\theta'$-related
to $e_0^l$ [$e_0^r$]. Recall that $C_l'$ and $C_r'$ are cut-sets of
cuts of $G'$. Thus, they induce cut-sets of $G$ denoted by $C_l$ and
$C_r$. Generally ''left`` and ''right`` \wrt an edge $e_i$ is from the
perspective of standing on $e_i$ and looking into $F_i$. Finally,
$C_{\tau} = \{e \in E \mid e_0~\tau~e\}$.
\end{definition}

\subsection{Embedding of cuts}
\label{subsec:plane_embed}
In this section we first represent a cut of $G$ through $e_0$ with
cut-set $C$ by a simple path or simple cycle $\gamma(C)$ in the
\emph{line graph} (sometimes referred to as \emph{edge graph}) $L_G$
of $G$. We then embed the edges of $L_G$ that we need for representing
cuts. In particular, all $\gamma(C)$ turn into simple non-closed or
closed curves.

\begin{definition}[$L_G(V^L, E^L)$, cut $\gamma(C)$]
\label{def:LG}
$L_G=(V^L, E^L)$ denotes the line graph of $G$, \ie $V^L = E$. Using
Notation~\ref{not:C}, we define $\gamma(C)$ to be the path in $L_G$
whose edge set is

\begin{equation}
E^L(C) = \{\{e_{i-1}, e_i\}\}
\end{equation}
\end{definition}

If $C$ is non-cyclic [cyclic], $\gamma(C)$ is a maximal simple path
[simple cycle] in $L_G$.

An edge $\{e, \hat{e}\}$ of $L_G$ can be part of $\gamma(C)$ for some
$C$ only if there exists a face $F$ of $G$ such that $e, \hat{e} \in
E(F)$. To embed such an edge we proceed basically as in
Section~\ref{sub:embed}. Let $p$ and $\hat{p}$ denote the midpoints of
$e$ and $\hat{e}$, respectively. Furthermore, let $F^r$ denote a
regular polygon with the same number of sides as $F$, and let
$\hslash:\closu{F^r} \mapsto \closu{F}$ be a homeomorphism (recall
that $\closu{F} = F \cup E(F)$). We embed $\{e, \hat{e}\}$ as
$\hslash(L)$, where $L$ is the line segment between $\hslash^{-1}(p)$
and $\hslash^{-1}(\hat{p})$. Thus, the vertices of embedded
$\gamma(C)$ are all midpoints of edges of $G$, and embedded
$\gamma(C)$ is a curve that subdivides $D(G)$ into two connected
components (for $D(G)$ see Definition~\ref{def:domain}).

\subsection{Restrictive conditions for convex cuts}
\label{subsec:plane_necessary}
Any cut-set of a convex cut through $e_0$ must be contained in $C_l
\cup C_r \cup C_{\tau}$. This follows from
Theorem~\ref{theorem_subdiv}, \ie the fact that for any $e_i$ in $C$
it must hold that either $e_0~\tau~e_i$ or that there exists a child
of $e_0$ and a child of $e_i$ that are $\theta'$-related.

The next lemma tells us that, on a local level, we have to deal only
with $\theta'$ and not with $\tau$.

\begin{lemma}
\label{lemma:tau1}
If $e_{i-1}~\tau~e_i$, then $C$ cannot be the cut-set of a convex cut.
\end{lemma}

\begin{proof}
Let $e_{i-1} = \{u_{i-1}, v_{i-1}\}$, $e_i = \{u_i, v_i\}$. Without
loss of generality we assume that $u_{i-1}$ is on the same side of the
convex cut as $u_i$ and that $v_{i-1}$ is on the same side of the
convex cut as $v_i$ (see Figure~\ref{fig:tau}a). Then
$e_{i-1}~\tau~e_i$ and $e_{i-1}, e_i \in E(F_{i-1})$ imply that any
shortest path from $u_{i-1}$ to $v_i$ intersects any shortest path
from $v_{i-1}$ to $u_i$ at a vertex that we denote by $w$. Without
loss of generality we may assume that $w$ is on the same side of the
convex cut as $u_{i-1}$. Due to $e_{i-1}~\tau~e_i$ we have
$d_G(w, u_i) = d_G(w, v_i)$. Thus, there exists a shortest path from
$v_i$ via $w$ to $v_{i-1}$ that starts and ends on the side of
$v_{i-1}$, but contains the vertex $w$, which is on the side of
$u_{i-1}$. Hence the cut cannot be convex.
\end{proof}

The observation below will lead to more restrictive conditions for
convex cuts.

\begin{observation}
\label{obs:plane}
The case distinction in the proof of Theorem~\ref{theorem_subdiv}
yields the following for plane graphs. If a child $e_i'$ of $e_i$ is
$\theta'$-related to a child $e_j'$ of $e_j$ with $j \neq i$, then
exactly one of the next two cases holds.

\begin{itemize}
\item[a)] $\theta'$ induces a one-to-one correspondence between the
  two children of $e_i$ and the two children of $e_j$ (see
  Figure~\ref{fig:notation-children}a). If $e_i'$ is the left
  [right] child of $e_i$, then $e_j'$ is the right [left] child
  of $e_j$.
\item[b)] The pair $e_i'$, $e_j'$ is the only pair of
  $\theta'$-related children (see
  Figure~\ref{fig:notation-children}b). In particular, $e_i'$
  and $e_j'$ must be on the same side of the cut, and there exists a
  shortest path from the end vertex of $e_i$ that is also the
  end vertex of $e_i'$ via $e_i$ to the end vertex of $e_j$
  that is not an end vertex of $e_j'$. If $e_i'$ is the left
  [right] child of $e_i$, then $e_j'$ is the left [right] child
  of $e_j$.
\end{itemize}
\end{observation}

All we know about the cut $(V_1, V_2)$ in the next lemma is that a
pair of edges has certain children that are $\theta'$ related. Still,
$(V_1, V_2)$ tells us that certain convex cuts cannot exist.

\begin{lemma}
\label{lemma:deg2}
Let $C$ be the cut-set of an embedded cut $(V_1, V_2)$ such that the
left [right] child of $e_i$ is $\theta'$-related to a child of $e_j$
for some $i, j$ with $j > i$. Then there exists no (embedded) convex
cut with $e_i$ in its cut-set that runs right [left] of $(V_1, V_2)$.
\end{lemma}

\begin{proof}
Without loss of generality we assume that the left child of $e_i$,
denoted by $e^l_i$, is $\theta'$-related to a child of $e_j$, denoted
by $e'_j$. We denote the left and right end vertex of $e_i$ [$e_j$] by
$u_i$ and $v_i$ [$u_j$ and $v_j$], respectively. Due to
Observation~\ref{obs:plane}, one of the two following cases must hold.

\begin{enumerate}
\item $d_{G'}(u_i, u_j) = d_{G'}(v_i, v_j) =:k$, and $e'_j$ is the
  right child of $e_j$. This is the case illustrated in
  Figure~\ref{fig:tau}b. The shortest path from $u_i$ to $v_j$ in $G'$
  cannot be shorter than $k+2$, because this would entail $d_{G'}(u_i,
  v_j) = k$ and thus $v_j \in W_{u_i, v_i}$, a contradiction to
  $e^l_i~\theta'e'_j$. Hence, there exists a shortest path $P'$ in
  $G'$ from $u_i$ via $v_i$ to $v_j$ (with length $k+2$). A cut with
  $e_i$ in its cut-set that runs right of $(V_1, V_2)$ is crossed
  twice by $P'$. Hence the cut is not convex.
\item $d_{G'}(u_i, u_j) = d_{G'}(v_i, v_j) + 2$, and $e'_j$ is the
  left child of $e_j$. This is the case illustrated in
  Figure~\ref{fig:tau}c. From $e^l_i~\theta'~e'_j$ follows again that
  there exists a shortest path $P'$ in $G'$ from $u_i$ via $v_i$ to
  $v_j$ (with length $k+2$), and the claim follows as in the item
  above.
\end{enumerate}

\end{proof}

\begin{figure}
\begin{center}
(a) \includegraphics[height=0.23\columnwidth]{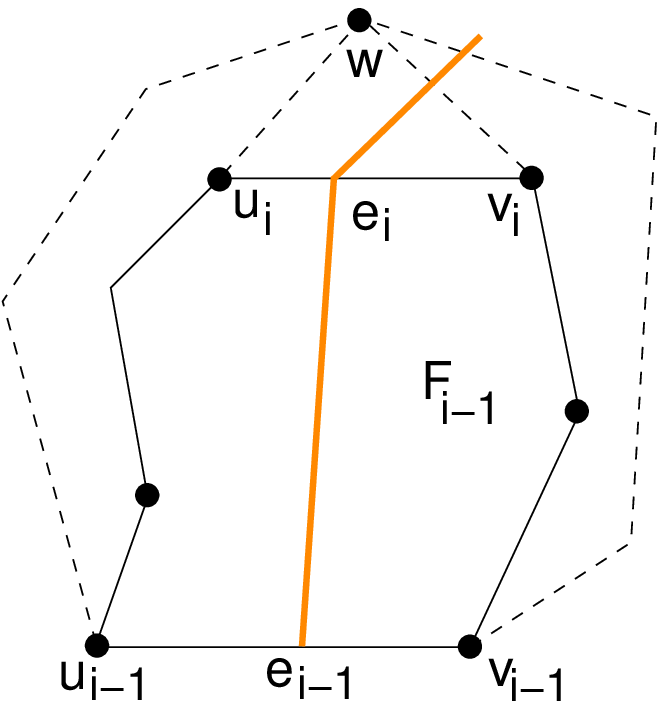}\qquad
(b) \includegraphics[height=0.20\columnwidth]{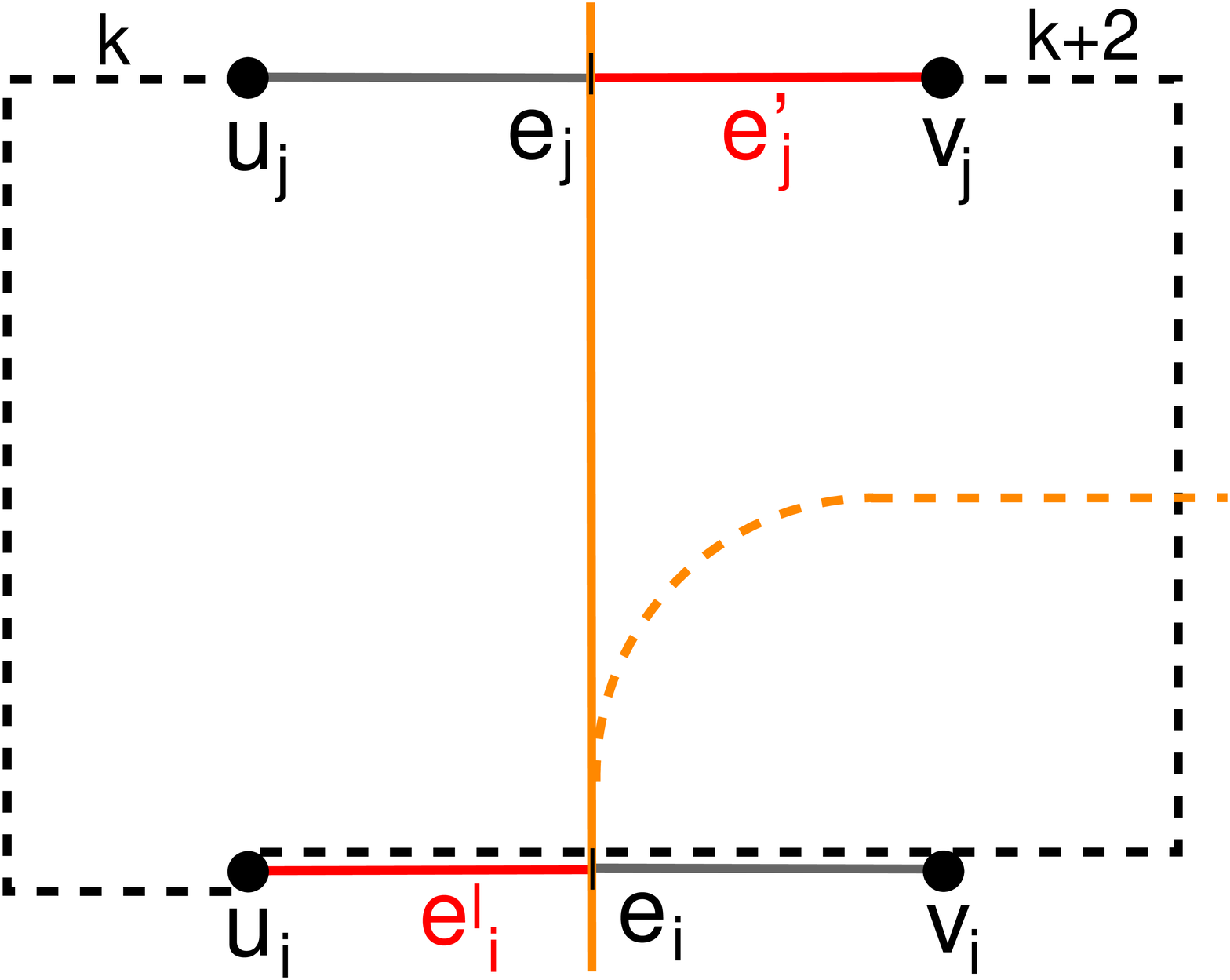}\qquad
(c) \includegraphics[height=0.20\columnwidth]{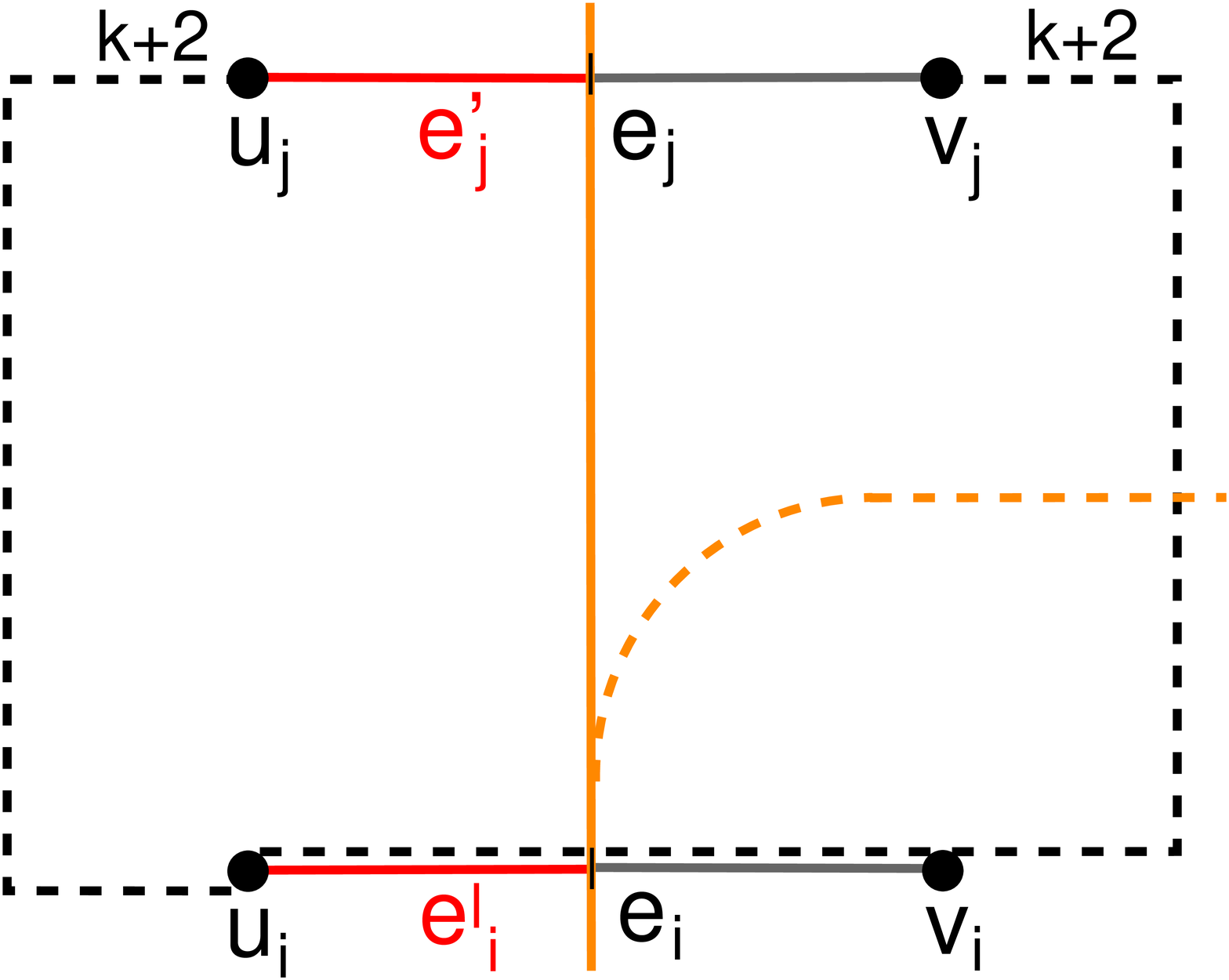}
\caption{(a) Illustration to proof of Lemma~\ref{lemma:tau1}. The
  dotted black edges indicate shortest paths in $G$, and the orange
  zigzag line indicates a non-convex cut. (b, c) Illustrations to
  proof of Lemma~\ref{lemma:deg2}. Curved gray lines indicate paths on
  $E(F_{i-1})$, and gray line segments indicate children of edges on
  $E(F_{i-1})$. The red children are $\theta'$-related. The cut with
  cut-set $C$ is indicated by the solid orange curve, and the cut
  indicated by the dashed orange curve cannot exist.}
\label{fig:tau}
\end{center}
\end{figure}

We will now see that embedded $C_l$ and $C_r$ border all embedded
convex cuts through $e_0$.

\begin{proposition}~\vspace{-0.5cm}\\
\label{prop:frame}
\begin{enumerate}
\item The embedded cut with cut-set $C_l$ runs on the right side of
  the embedded cut with cut-set $C_r$ (except on $C_l \cap C_r$, where
  the embedded cuts touch).
\item Any embedded convex cut runs between the embedded cut with
  cut-set $C_l$ and the embedded cut with cut-set $C_r$, \ie no part
  of the convex cut runs right of $C_l$ or left of $C_r$.
\end{enumerate}
\end{proposition}

\begin{proof}~\vspace{-0.5cm}\\
\begin{enumerate}
\item We use Observation~\ref{obs:plane}b: for any $e^l \in C_l
  \setminus C_r$ there exists a shortest path $P$ in $G'$ from the
  left end vertex of $e_0$ via $e_0$ and the right end vertex of $e_0$
  towards the right end vertex of $e^l$. The assumption that parts of
  $C_l$ run on the left side of $C_r$ lead to a contradiction. Indeed,
  this would entail that there exists a shortest path $P$ as above
  which also crosses $C_r$ via an edge $e^r \in C_r$, \ie $P$ contains
  a shortest path from the left end vertex of $e_0^r$ via $e_0^r$ and
  the right end vertex of $e_0^r$ (which equals the right end vertex
  of $e_0$) and further on via the right end vertex of $e^r$ to the
  left end vertex of $e^r$ --- a contradiction to $e^r \in C_r$.
\item A consequence of a special case of Lemma~\ref{lemma:deg2}, \ie
  the case $i=0$.
\end{enumerate}
\end{proof}

The following proposition reveals that the edges of $G$ that
are not in $C_l \cup C_r$, \ie the edges in $C_{\tau}$, may serve
as unique sequences of stepping stones for convex cuts that move from
$C_l$ to $C_r$ or vice versa.

\begin{proposition}
\label{prop:stepping}
Let $C = (e_0, \dots, e_{\vert C \vert -1})$ be the cut-set of a
convex cut through $e_0$. Then the following holds. If $e_{i-1} \in
C_l$ [$e_{i-1} \in C_r$] and $e_i \in C_{\tau}$, then there exists $j > i$
such that $e_i, \dots, e_{j-1} \in C_{\tau}$ and $e_j \in
C_r$ [$e_j \in C_l$]. Moreover, any cut-set of a convex cut through
$e_0$ that coincides with $C$ on $e_0, e_1, \dots, e_i$ must coincide
with $C$ on $e_0, e_1, \dots, e_j$.
\end{proposition}
\
\begin{proof}
Without loss of generality we assume $e_{i-1} \in C_l$.
\begin{enumerate}
\item Let $e_0 = \{u_0, v_0\}$ and $e_i = \{u_i, v_i\}$, let $P_u$
  [$P_v$] be a shortest path from $u_i$ [$v_i$] to $u_0$ in $G$, and
  let $e^u$ [$e^v$] be the first edge on $P_u$ [$P_v$] (see
  Figure~\ref{fig:loops}a). Then $e^u, e^v \in C_r$.

  Without loss of generality we show that $e^u \in C_r$. Let $w_0$
  [$w_i$] denote the vertex of $G'$ that subdivides $e_0$ [$e^u$]. To
  prove $e^u \in C_r$, it suffices to show $\{w_i, u_i\} \theta' \{u_0,
  w_0\}$. Indeed, by definition of $\tau$, we have that the length of
  $P_u$ equals $d_G(u_i, v_0)$. Since the degrees of $w_i$ and $w_0$
  are two, a shortest path from $w_i$ to $w_0$ runs via $u_0$ or via
  $v_0$. In both cases the distance is $2 d_{G'}(u_i, v_0)$. Thus,
  $\{w_i, u_i\}$ is $\theta'$-related to $e_0^r$, \ie $e^u \in C_r$.
\item The following case distinction yields $e^v = e_j$ for some
  $j > i$ (see Figure~\ref{fig:loops}b).

\begin{enumerate}
\item The embedded convex cut with cut-set $C$ crosses embedded
  $C_r$. This case cannot occur due to Lemma~\ref{lemma:deg2}.
\item $C$ contains $e^u$. Then $\{w_i, u_i\}$ is a left child of $e^u$
  \wrt $C$. Since $e_0^r$ is a right child of $e_0$ \wrt $C$, and
  $e_0^r$ is $\theta'$-related to $\{w_i, u_i\}$ (see item 2a),
  Observation~\ref{obs:plane}a yields that $e_0^l$ and the right child
  of $e^u$ are $\theta'$-related, too. This is a contradiction to $e^u
  \in C_r$ (see item 1).
\item The remaining case is that $C$ contains $e^v \in C_r$, \ie $e^v
  = e_j$ for some $j > i$.
\end{enumerate}

\item It remains to show that the extension from $e_i$ to $e_j$ is
  unique (see Figure~\ref{fig:loops}c). Indeed, let $y_i$ be the end
  vertex of $e_j$ that is not $v_i$. The path from $u_i$ to $y_i$ via
  $v_i$ has length two and contains two edges of $C$. Due to the cut
  being convex, there exists an edge $\hat{e}$ from $u_i$ to $y_i$,
  \ie the edges $e_i$, $e_j$ and $\hat{e}$ form a triangle.

  If $d_G(x, v_i) < d_G(x, y_i)$, then $x$ must be on the same side of
  the cut as $v_0$. Indeed, assume that $x$ is on the other side, \ie
  the side of $u_0$. Then, due to $e_0~\tau~e_i$, \ie $d_G(u_0,
  u_i)=d_G(u_0, v_i)=d_G(v_0, u_i)=d_G(v_0, v_i)$, there exists a
  shortest path from $x$ via $v_i$ to $v_0$ that crosses the convex
  cut twice, a contradiction.

  If $d_G(x, v_i) \geq d_G(x, y_i)$, then $x$ must be on the same side
  of the convex cut as $u_0$. Indeed, assume that $x$ is on the other
  side, \ie the side of $v_0$. Then, due to $e_0~\tau~e_i$, there
  exists a path from $x$ via $y_i$ and $u_0$ to $v_0$ that is not
  longer than alternative paths of $x$ via $u_i$ or $v_i$ to
  $v_0$. The path via $y_i$ thus is a shortest path from $x$ to $v_0$
  that crosses the convex cut twice, a contradiction.

  To summarize, the side of any vertex $x$ in the triangle is unique
  \ie the extension from $e_i$ to $e_j$ is unique.
\end{enumerate}
\end{proof}

\begin{figure}[H]
\begin{center}
\includegraphics[height=0.27\columnwidth]{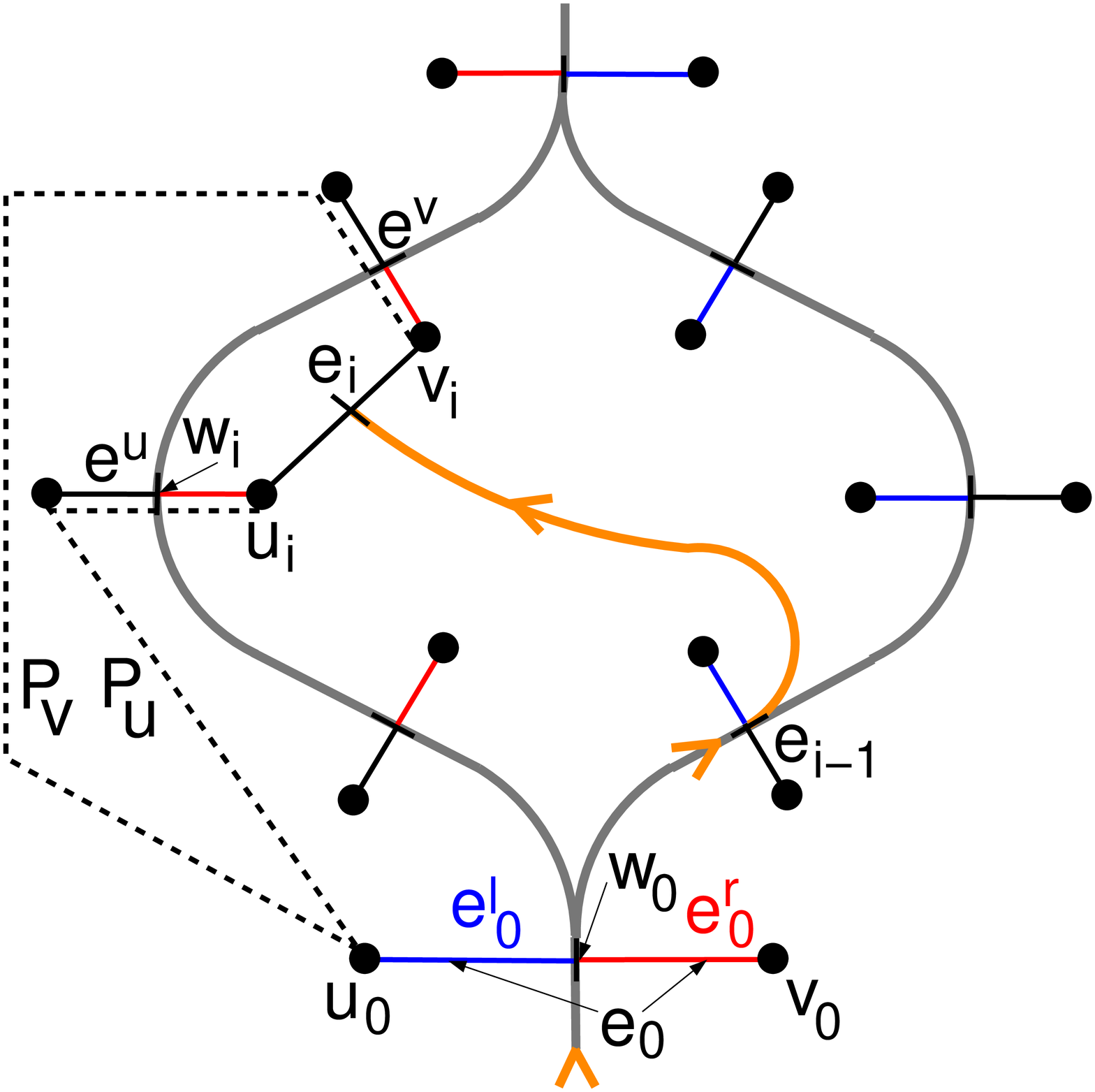}\qquad
\includegraphics[height=0.27\columnwidth]{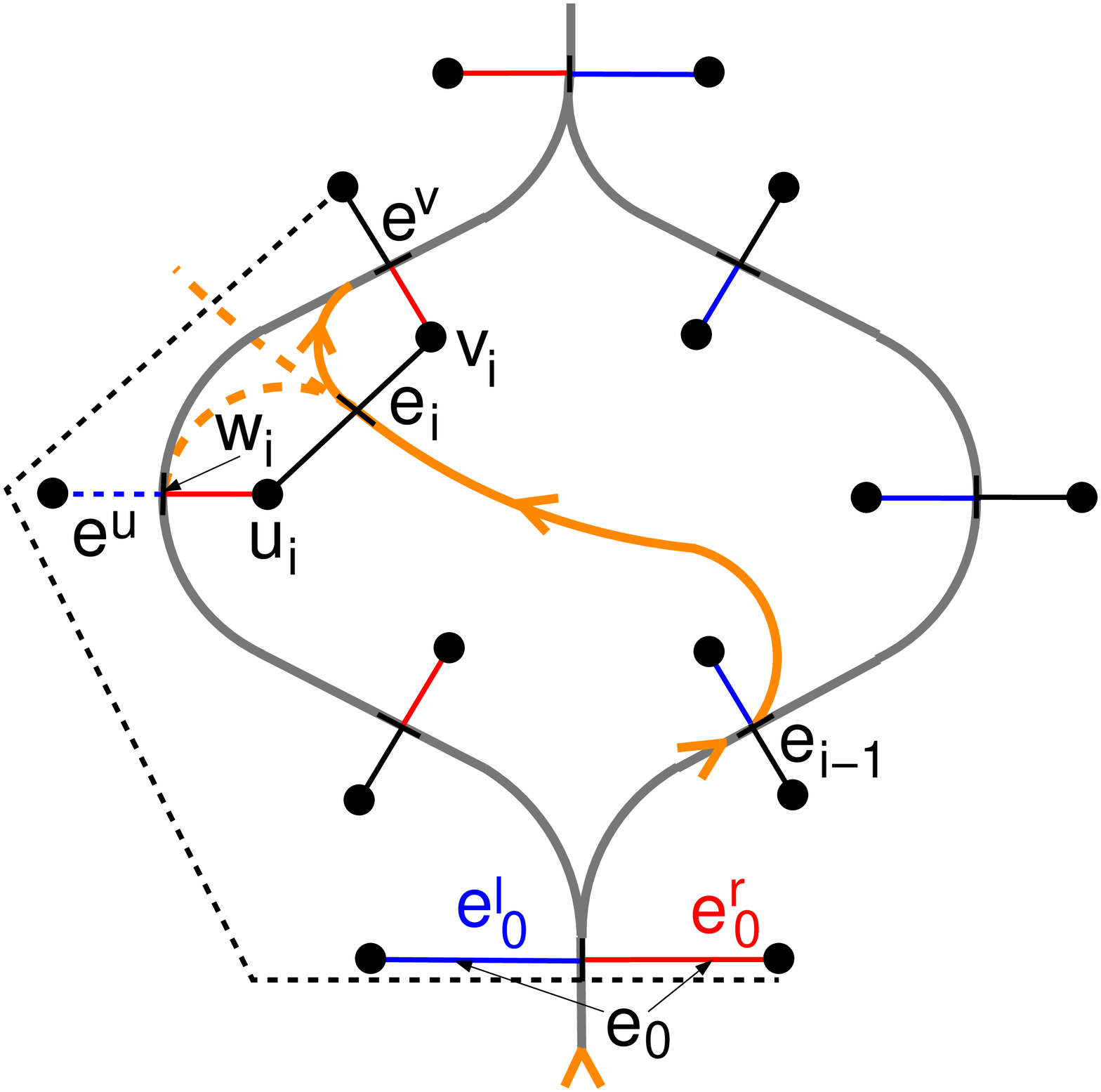}\qquad
\includegraphics[height=0.27\columnwidth]{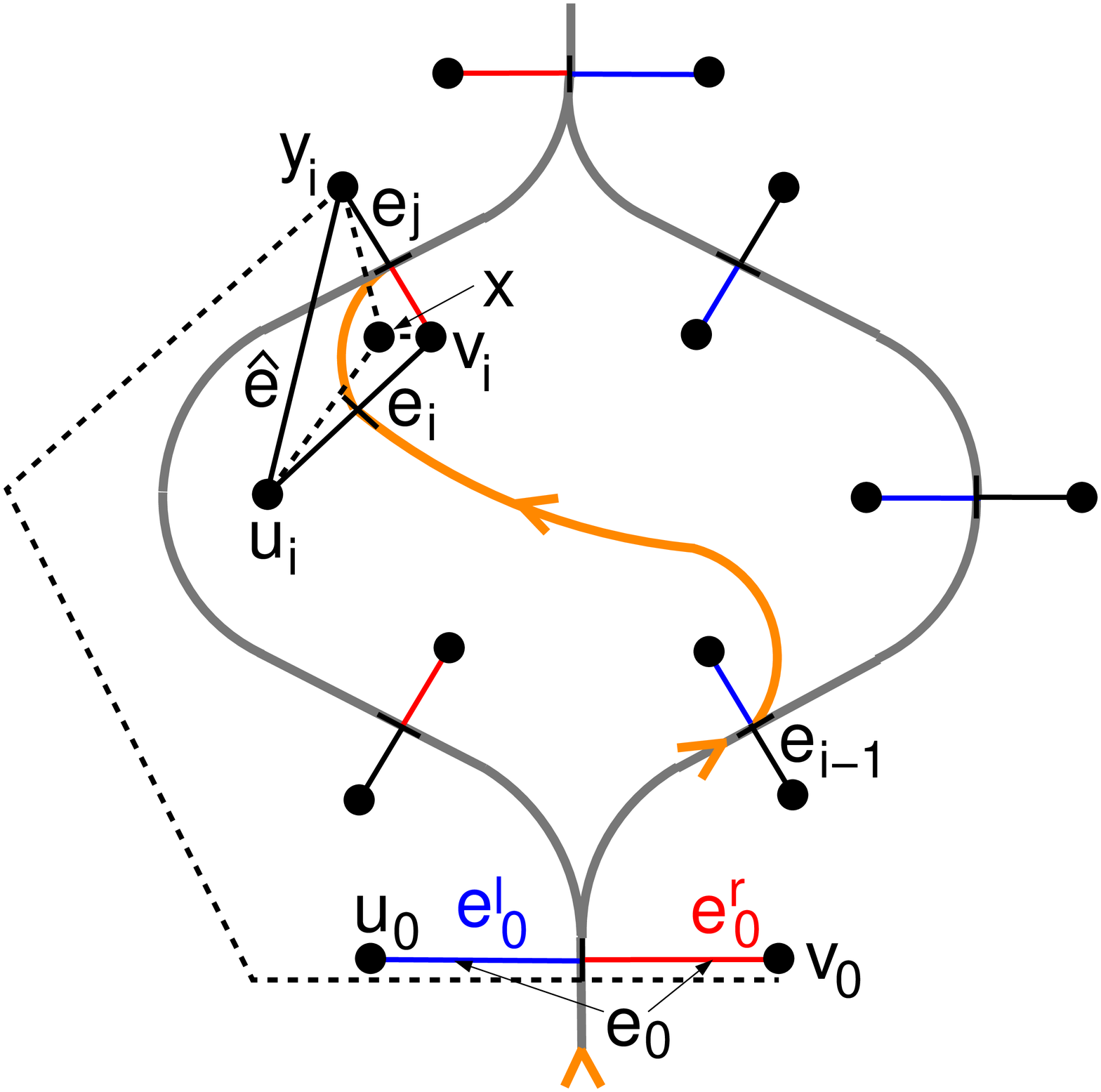}
\caption{Illustrations to the proof of Lemma~\ref{prop:stepping}. The
  cuts defined by $C_l$ and $C_r$ are indicated by the two gray curves
  that bifurcate at $e_0$ and merge at the top. Vertices of $G$ [$G'$]
  are marked as filled circles [bars], and edges of $G$ are shown as
  solid line segments between filled circles, possibly consisting of
  two colors indicating the children. Colors of children indicate
  membership to $C_l'$ and $C_r'$. The orange arrowheads and the
  orange line indicate the cut defined by $C$. The black edge $e_i$ is
  $\tau$-related to $e_0$. (a) The dashed black zigzag lines indicate
  shortest paths $P_u$ and $P_v$ from the left end vertex $u_0$ of
  $e_0$ to the left end vertex $u_i$ of $e_i$, and from $u_0$ to the
  right end vertex $v_i$ of $e_i$, respectively.  (b) Dashed orange
  lines indicate potentially convex cuts that turn out to be
  non-convex because of the shortest paths indicated by the dashed
  black lines (see the proof of Lemma~\ref{prop:stepping}). (c) There
  exists an edge $\hat{e}=\{u_i, y_i\}$. Furthermore, if $d_G(x, v_i)
  \geq d_G(x, y_i)$, then there exists a shortest path from $x$ via
  $y_i$ and $u_0$ to $v_0$.}
\label{fig:loops}
\end{center}
\end{figure}

\subsection{Intersection pattern of embedded convex cuts}
\label{subsec:IPs}
In Section~\ref{subsec:plane_embed} we represented a cut of $G$
through $e_0$ by an embedded simple path or simple cycle $\gamma(C)$
in the line graph $L_G$ of $G$. In this section we study the
intersection pattern of a pair of embedded convex cuts of $G$ through
$e_0$. More formally, if $C$ and $\hat{C}$ are cut-sets of convex cuts
of $G$ through $e_0$, we study the patterns in $\RR^2$ that are formed
by the curves $\gamma(C)$ and $\gamma(\hat{C})$.

The boundary of any face $F_L$ of $\gamma(C) \cup \gamma(\hat{C})$
constitutes a cyclical cut of $G$ (see
Figure~\ref{fig:touch-cross}). Let $v^L \in V^L$ be a vertex on the
boundary of $F_L$. The vertex $v^L$ is the midpoint of an edge $e \in
C$. In particular, one end vertex of $e$ must be contained in the
interior of $F_L$. Thus we have proven the following.

\begin{lemma}
\label{lemma:vertex-inside}
Any face of $\gamma(C) \cup \gamma(\hat{C})$ contains at least one vertex of
$G$. 
\end{lemma}


Lemma~\ref{lemma:tau1} and the case of Lemma~\ref{lemma:deg2} in which
$e_i$ and $e_j$ sit on the boundary of the same face of $G$ yield
Proposition~\ref{prop:deg2} (see also Figures~\ref{fig:tau}b,c). It
states that all embedded convex cuts through $e_0$ which have reached
a face $F$ on the midpoint of an edge on $E(F)$ can cut through $F$ in
at most two ways.

\begin{proposition}
\label{prop:deg2}
Let $C = (e_0, \dots e_{\vert C \vert -1})$ be the cut-set of a convex
cut of $G$ through $e_0$. Then the following holds. For all $e_i$ in
$C$ there exists $f_i \in E(F_{i-1})$ (possibly $e_i = f_i$) such that
$\hat{e}_j \in \{e_i\} \cup \{f_i\}$ for the cut-set $\hat{C} = (e_0,
\hat{e}_1, \dots \hat{e}_{\vert \hat{C} \vert -1})$ of any convex cut
with $\hat{e}_{j-1} = e_{i-1}$ and $\vert \{e_0, \hat{e}_1, \dots
\hat{e}_{j-1}\} \cap F_{i-1}\vert = 1$.
\end{proposition}

\begin{figure}
\begin{center}
(a) \includegraphics[height=0.32\columnwidth]{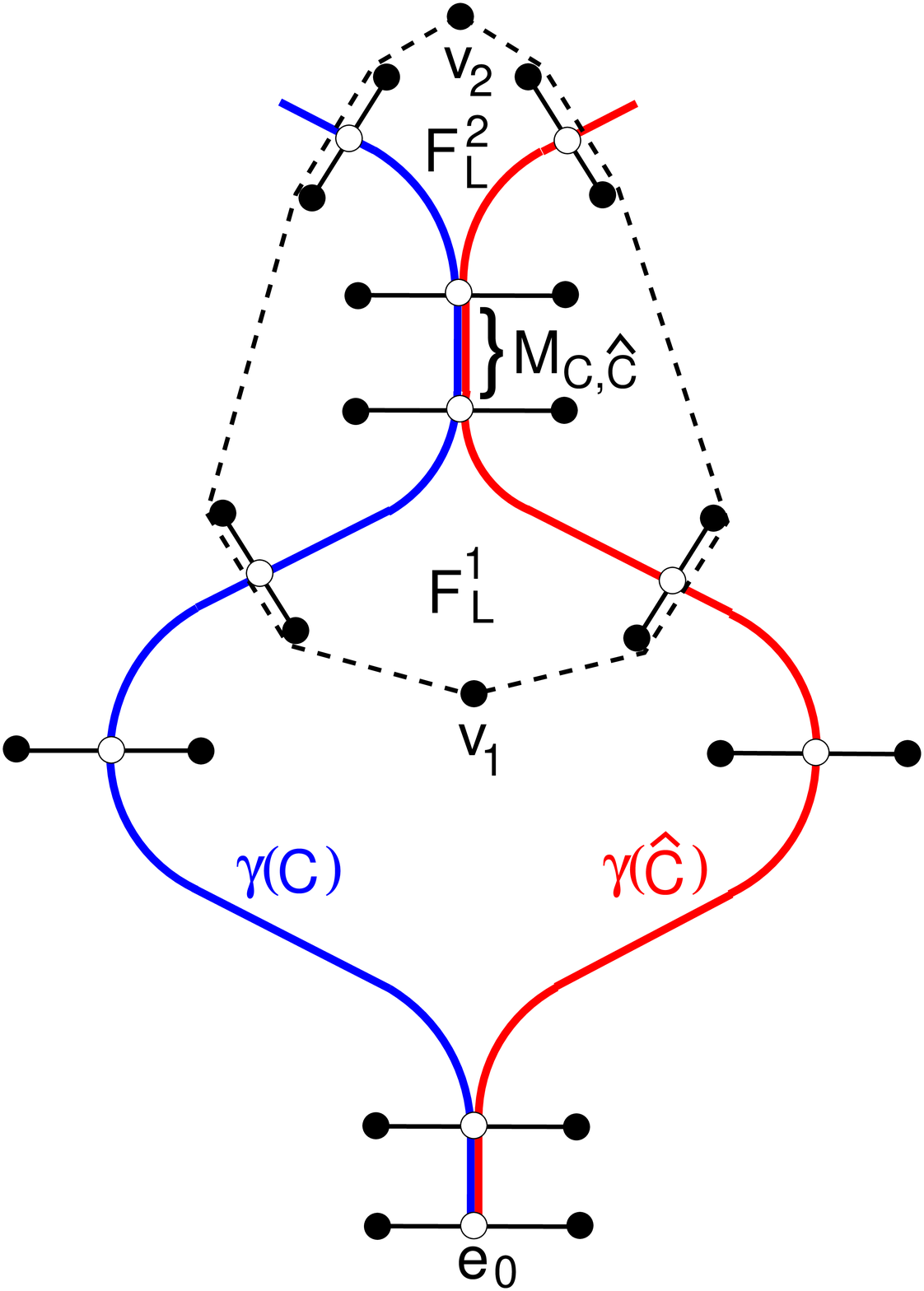}\qquad
(b) \includegraphics[height=0.32\columnwidth]{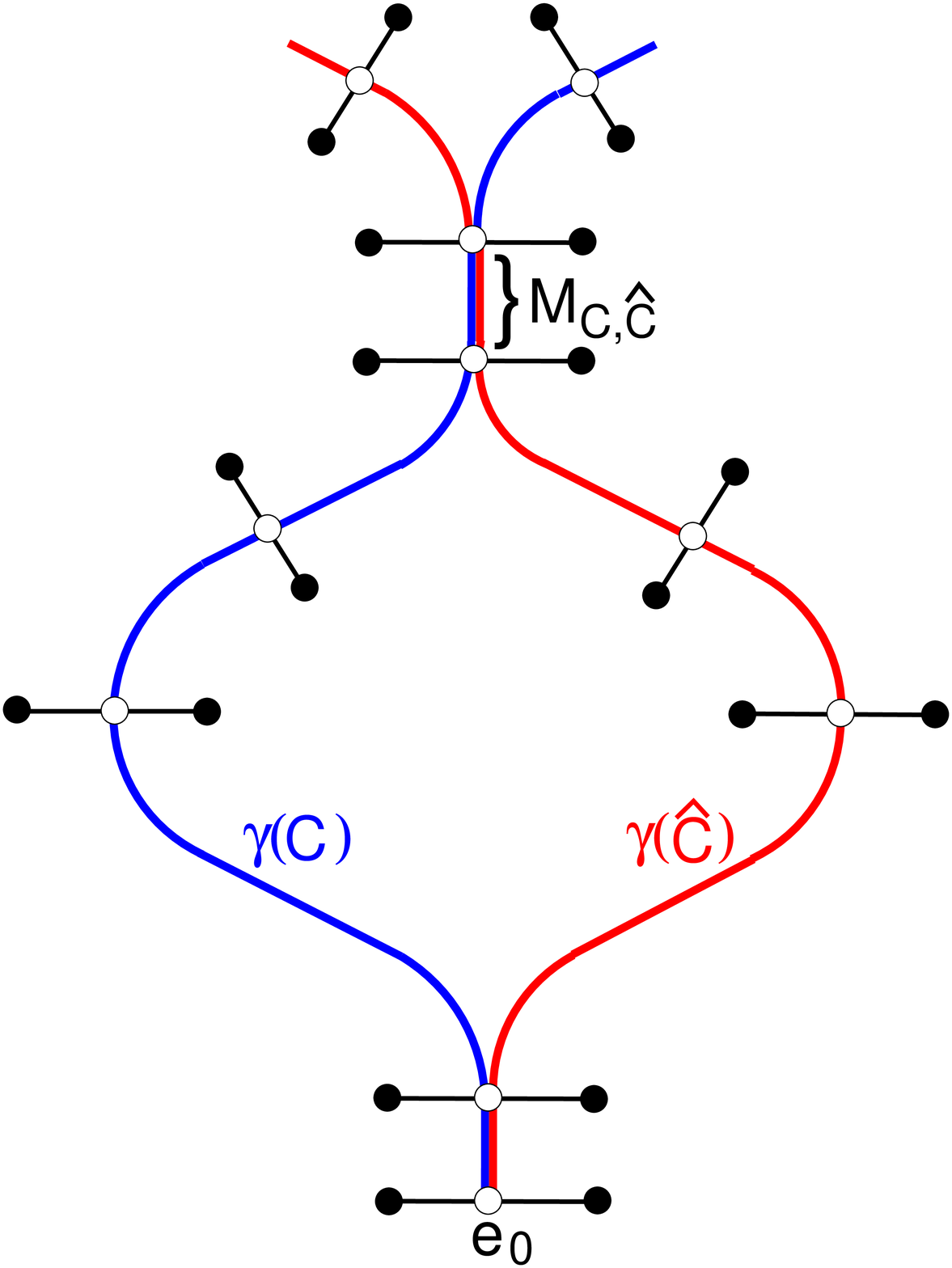}\qquad
(c) \includegraphics[height=0.32\columnwidth]{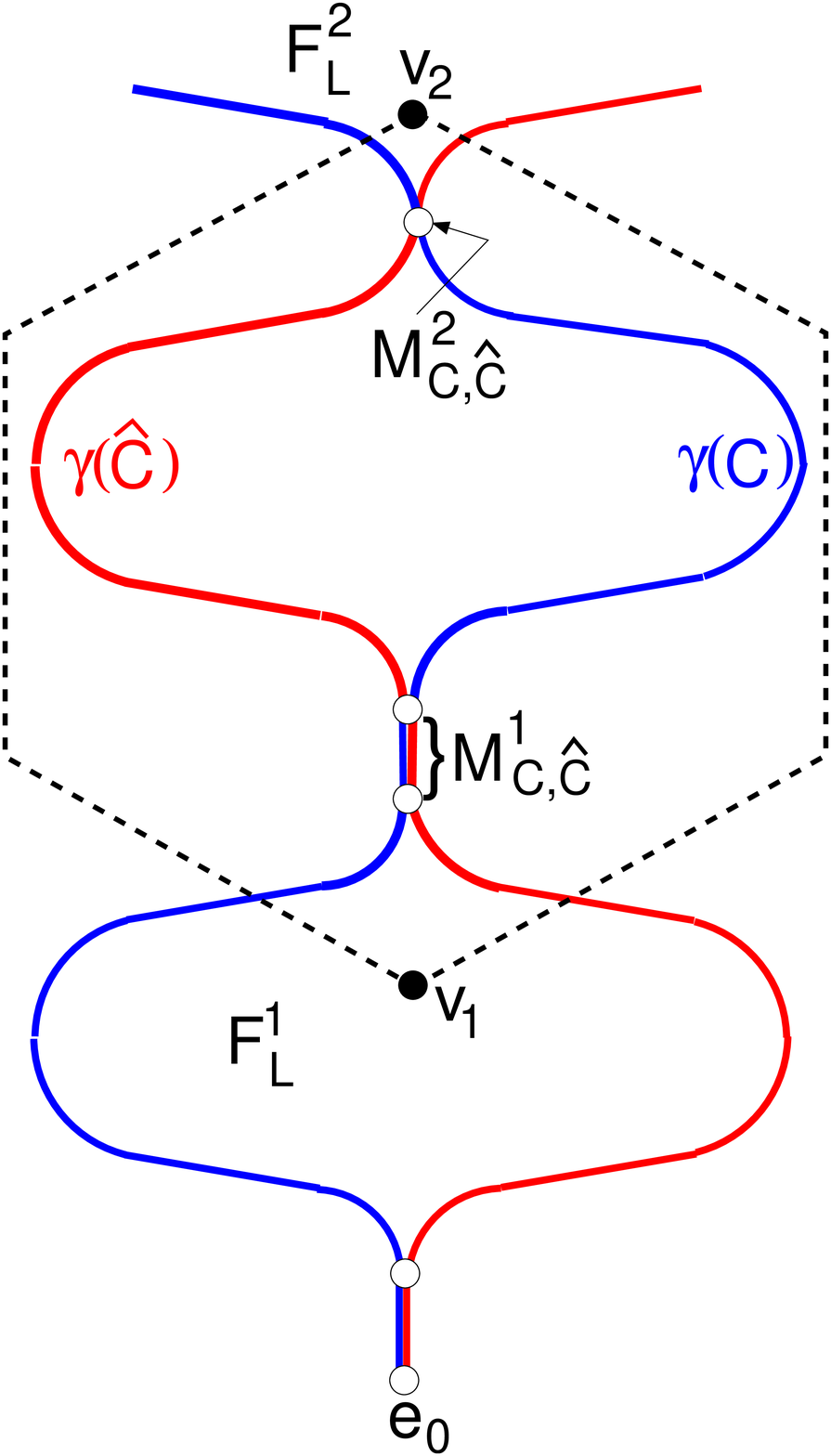}
\caption{(a) Illustration to proof of
  Proposition~\ref{prop:touch-cross}.  The embedded cuts $\gamma(C)$
  and $\gamma(\hat{C})$ are indicated by the blue and red curves,
  respectively. Vertices and edges of $G$ are marked as filled circles
  and solid black lines, respectively. Shortest paths in $G$ are shown
  as dotted black lines. Note that shortest paths in $G$ can use only
  vertices marked as filled circles. Vertices of $S$ are shown as open
  circles. (a, b, c) Illustration to the proof of
  Proposition~\ref{prop:touch-cross}. (a) $\gamma(C)$ touches
  $\gamma(\hat{C})$ on the maximal common curve $M_{C, \hat{C}}$ of
  $\gamma(C)$ and $\gamma(\hat{C})$. (b) $\gamma(C)$ crosses
  $\gamma(\hat{C})$ on $M_{C, \hat{C}}$. (c) $\gamma(C)$ and
  $\gamma(\hat{C})$ have crossings $M^1_{C, \hat{C}}$ and $M^2_{C,
    \hat{C}}$.}
\label{fig:touch-cross}
\end{center}
\end{figure}

\begin{definition}[$\gamma(C)$ touches $\gamma(\hat{C})$, $\gamma(C)$ crosses $\gamma(\hat{C})$, overlap, crossing $M_{C, \hat{C}}$]
\label{def:touch-cross}
Let $C$ and $\hat{C}$ be cut-sets of cuts of $G$ through $e_0$. We say
that $\gamma(C)$ touches [crosses] $\gamma(\hat{C})$ on the maximal common
curve $M_{C, \hat{C}}$ of $\gamma(C)$ and $\gamma(\hat{C})$ if the part of
$\gamma(C)$ directly before $M_{C, \hat{C}}$ is on the same side [on the
  other side] of $\gamma(\hat{C})$ as the part of $\gamma(C)$ directly after
$M_{C,\hat{C}}$. The curve $M_{C, \hat{C}}$ is called an overlap of
$\gamma(C)$ and $\gamma(\hat{C})$. If $\gamma(C)$ crosses $\gamma(\hat{C})$ on
$M_{C, \hat{C}}$, we refer to $M_{C, \hat{C}}$ as the crossing of
$\gamma(C)$ and $\gamma(\hat{C})$.
\end{definition}

For examples of touching and crossing cuts see
Figure~\ref{fig:touch-cross}. The following proposition describes the
intersection pattern of a pair of embedded convex cuts of $G$ through
$e_0$.

\begin{proposition}
\label{prop:touch-cross}
Let $C, \hat{C}$ be cut-sets of convex cuts of $G$ through $e_0$. Then
$\gamma(C)$ cannot touch $\gamma(\hat{C})$, and $\gamma(C)$ and $\gamma(\hat{C})$
can have at most one crossing.
\end{proposition}

\begin{proof}~\vspace{-0.5cm}\\
\begin{itemize}
\item Assume that $\gamma(C)$ touches $\gamma(\hat{C})$. Let $F^1_L$ and
  $F^2_L$ be the faces formed by the parts of $\gamma(C)$ and
  $\gamma(\hat{C})$ before and after $M_{C, \hat{C}}$, respectively (see
  Figure~\ref{fig:touch-cross}a). Lemma~\ref{lemma:vertex-inside}
  yields that there exist vertices $v_1 \in V \cap F^1_L$ and $v_2 \in
  V \cap F^2_L$. Any shortest path from $v_1$ to $v_2$ either crosses
  $\gamma(C)$ or $\gamma(\hat{C})$ twice, a contradiction to $C$ and
  $\hat{C}$ being cut-sets of convex cuts.
\item Assume that $\gamma(C)$ and $\gamma(\hat{C})$ have crossings $M^1_{C,
  \hat{C}} \neq M^2_{C, \hat{C}}$, and that there is no crossing
  between $M^1_{C, \hat{C}}$ and $M^2_{C, \hat{C}}$. Let $F^1_L$ and
  $F^2_L$ be the faces formed by the parts of $\gamma(C)$ and
  $\gamma(\hat{C})$ before $M^1_{C, \hat{C}}$ and after $M^2_{C,
    \hat{C}}$, respectively (see
  Figure~\ref{fig:touch-cross}c). Lemma~\ref{lemma:vertex-inside}
  yields that there exist vertices $v_1 \in V \cap F^1_L$ and $v_2 \in
  V \cap F^2_L$. As in the previous item, any shortest path from $v_1$
  to $v_2$ either crosses $\gamma(C)$ or $\gamma(\hat{C})$ twice, a
  contradiction to $C$ and $\hat{C}$ being cut-sets of convex cuts.
\end{itemize}
\end{proof}

\subsection{Upper bound on number of convex cuts through $e_0$}
\label{subsec:numbers}
To find an upper bound on the number of convex cuts of $G$ through
$e_0$ we start by assuming that there exists at least one such cut with
cut-set $\hat{C}$. The necessary conditions for convex cuts in
Section~\ref{subsec:plane_necessary} and
Proposition~\ref{prop:touch-cross} impose constraints on the other
candidates for convex cuts through $e_0$. In particular,
Proposition~\ref{prop:touch-cross} implies the following. The first
overlap of $\gamma(C)$ and $\gamma(\hat{C})$ is always the one that
contains the midpoint of $e_0$. It always exists. If there is a second
overlap, and this second overlap is not a crossing, it must be the
last overlap, and it must contain the midpoint of $e_{\vert C \vert
  -1}$. This follows from the fact that $\gamma(C)$ and
$\gamma(\hat{C})$ cannot touch. If the second overlap exists and
constitutes a crossing, there may or may not be another overlap. If
there exists such a third overlap, it must be the last one, and it
must contain $e_{\vert C \vert -1}$ (since $\gamma(C)$ and
$\gamma(\hat{C})$ cannot touch, and they cannot cross twice). Thus, we
have proven the following.

\begin{proposition}[At most three overlaps]
\label{prop:3over}
Let $C$ and $\hat{C}$ be cut-sets of convex cuts of $G$ through
$e_0$. Then $\gamma(C)$ and $\gamma(\hat{C})$ can have at most three
overlaps.
\end{proposition}

\begin{definition}[Fork and corresponding join of $\gamma(C)$ and $\gamma(\hat{C})$, detour on $\gamma(C)$]
\label{def:fork-join}
Let $C$ and $\hat{C}$ be cut-sets of cuts of $G$ through
$e_0$. Furthermore, let $M_1$ and $M_2$ be two consecutive overlaps of
$\gamma(C)$ and $\gamma(\hat{C})$, \ie there is no overlap of $\gamma(C)$ and
$\gamma(\hat{C})$ of $M_1$ and $M_2$. Then the last point of $M_1$
and the first point of $M_2$ are called fork and corresponding join
of $\gamma(C)$ and $\gamma(\hat{C})$. The sub-path of $\gamma(C)$ between
a fork and a corresponding join is called detour on $\gamma(C)$ around
$\gamma(\hat{C})$.
\end{definition}

\begin{proposition}[Unique detours]
\label{prop:detours}
Let $C$, $C^*$ and $\hat{C}$ be cut-sets of convex cuts of $G$ through
$e_0$, let $p^f$ and $p^j$ be a fork and a corresponding join of
$\gamma(C)$ and $\gamma(\hat{C})$, as well as of $\gamma(C^*)$ and
$\gamma(\hat{C})$. If $\{p^f\}$ is not a crossing, then the corresponding
detours on $\gamma(C)$ and on $\gamma(C^*)$ coincide.
\end{proposition}

\begin{proof}~\vspace{-0.5cm}\\
\begin{enumerate}
\item $p^f$ is the midpoint of an edge of $G$. Indeed, an edge of
  $\gamma(C)$, $\gamma(C^*)$ or $\hat{\gamma(C)}$ takes the form $\{f_{i-1},
  f_i\}$, where $f_{i-1}, f_i$ are edges of $G$. Here the indices
  reflect the order of the corresponding cut-sets (see
  Notation~\ref{not:C}). Moreover, there exists a face $F_{i-1}$ of
  $G$ with $f_{i-1}, f_i \in F_{i-1}$. The {\em embedded} edge
  $\{f_{i-1}, f_i\}$ is a curve from the midpoint of $f_{i-1}$ through
  the interior of $F_{i-1}$ to the midpoint of $f_i$. The embedding is
  such that two embedded edges that cross in a face of $G$ must cross
  at a single point (see Section~\ref{subsec:plane_embed}). Thus, the
  condition that the join $\{p^f\}$ is not a crossing implies that
  $p^f$ is the midpoint of an edge of $G$.
\item The first edge of the detour on $\gamma(C)$ around
  $\gamma(\hat{C})$ from $p^f$ to $p^j$ must coincide with the first
  edge of the detour on $\gamma(C^*)$ around
  $\gamma(\hat{C})$. Indeed, let $e^f$ be the edge with midpoint
  $p^f$. Let the successors of $e^f$ in the cut-sets $C$, $C^*$ and
  $\hat{C}$ be denoted by $s$, $s^*$ and $\hat{s}$,
  respectively. Since $p^f$ and $p^j$ are a fork and a corresponding
  join of $\gamma(C)$ and $\gamma(\hat{C})$, as well as of
  $\gamma(C^*)$ and $\gamma(\hat{C})$, the edges $\{e^f, s\}$, $\{e^f,
  s^*\}$ and $\{e^f, \hat{s}\}$ all go through the same face of
  $G$. Proposition~\ref{prop:deg2} yields that at least two of the
  three embedded edges must coincide. Thus, the first edge of the
  detour on $\gamma(C)$ around $\gamma(\hat{C})$ from $p^f$ and $p^j$
  must coincide with the first edge of the detour on $\gamma(C^*)$
  around $\gamma(\hat{C})$.

\item Analogous to Definition~\ref{def:equiv-cuts}, let $C^f_l$
  [$C^f_r$] be the set of edges of $G$ that have a child which is
  $\theta'$-related to the left [right] child of $e^f$. Without loss
  of generality we assume that $\gamma(C)$ runs right of
  $\gamma(\hat{C})$. Lemmas~\ref{lemma:tau1} and
  Proposition~\ref{prop:deg2} then yield that (i) the edge in $C$
  directly after $e^f$, denoted by $e$ is contained in $C^f_l$ and
  (ii) the edge in $\hat{C}$ directly after $e^f$, denoted by
  $\hat{e}$, is contained in $C^f_r$.

  We assume the opposite of the claim, \ie that $\gamma(C)$ and
  $\gamma(C^*)$ fork at or behind $e$ and join at or before
  $p^j$. Then Proposition~\ref{prop:frame} and
  Proposition~\ref{prop:stepping} imply that one of the embedded cuts
  continues on $C^f_l$, while the other one switches from $C^f_l$ to
  $C^f_r$. If the switching embedded cut hits $p^j$ as soon as it
  reaches $C^f_r$, the other embedded cut has missed $p^j$, a
  contradiction to $p^j \in \gamma(C)$ and $p^j \in \gamma(C^*)$. If
  the switching embedded cut does not hit $p^j$ as soon as it reaches
  $C^f_r$, the embedded cut $\gamma(\hat{C})$ has already switched
  from $C^f_r$ to $C^f_l$, and must thus have hit $\gamma(C)$ or
  $\gamma(C^*)$ before $p^j$, a contradiction.
\end{enumerate}
\end{proof}

\begin{proposition}
\label{prop:upper_bound}
An upper bound for the number of convex cuts of $G$ through $e_0$ is
$\vert E \vert^4$.
\end{proposition}

\begin{proof}
We may assume that there exists a convex cut through $e_0$ with
cut-set $\hat{C}$.
\begin{enumerate}
\item {\em Number of the convex cuts $\gamma(C)$ of $G$ through $e_0$
  that do not cross $\hat{C}$.}  Proposition~\ref{prop:touch-cross}
  yields that there can be at most one fork and corresponding join of
  $\gamma(\hat{C})$ and $\gamma(C)$. The number of $\gamma(C)$ is thus
  bounded by the number of detours around
  $\gamma(\hat{C})$. Proposition~\ref{prop:detours} yields that the
  number of detours cannot surmount the number of forks of
  $\gamma(\hat{C})$ and $\gamma(C)$ times the number of corresponding
  joins of $\gamma(\hat{C})$ and $\gamma(C)$, \ie at most $\vert E
  \vert (\vert E \vert -1) /2$.
\item {\em Number of the convex cuts $\gamma(C)$ of $G$ through $e_0$
  that cross $\hat{C}$.} We first select a sub-path $M_{\hat{C}}$ of
  $\gamma(\hat{C})$ and determine the number of $\gamma(C)$ that cross
  $\gamma(\hat{C})$ on $M_{\hat{C}}$. Let $\gamma(C)$ be such a
  path. $\gamma(C)$ joins $\gamma(\hat{C})$ at the first point of
  $M_{\hat{C}}$, denoted by $p^j$. Using
  Proposition~\ref{prop:touch-cross} we get that $\gamma(C)$ coincides
  with $\gamma(\hat{C})$ between $e_0$ and the last point of
  $M_{\hat{C}}$, with the exception of at most one detour before
  $M_{\hat{C}}$. We already know that $p^j$ is the join of a
  detour. Using Proposition~\ref{prop:detours} we get that the number
  of detours before $M_{\hat{C}}$ is less than $\vert E \vert$. The
  same holds for the number of detours behind $M_{\hat{C}}$. Thus, the
  number of $\gamma(C)$ that cross $\gamma(\hat{C})$ on $M_{\hat{C}}$
  is less than $\vert E \vert^2$. The number of non-empty sub-paths
  $M_{\hat{C}}$ of $\gamma(\hat{C})$, in turn, amounts to $\vert E
  \vert (\vert E \vert -1) /2$. Hence, the number of $\gamma(C)$ that
  cross $\gamma(\hat{C})$ is less than $(\vert E \vert^4 - \vert E
  \vert^3) / 2$.
\end{enumerate}
The total number of convex cuts of $G$ through $e_0$ thus cannot
surmount $\vert E \vert^4$.
\end{proof}

\subsection{Algorithm for finding all convex cuts}
\label{subsec:algo}
We search for convex cuts of $G$ using a subgraph $S$ of $L_G=(V^L,
E^L)$ (for $L_G$ see Definition~\ref{def:LG}).

\begin{definition}[Search graph $S=(V_S, E_S)$] We set\\
$E_S=\{\{e, f\} \in E^L \mid \{e, f\} \in E(F) \mbox{~for some face
    $F$ of $G$ and $e'~\theta'~f'$ for children $e'$ of $e$ and $f'$
    of $f$}\}$. The search graph $S$ is the subgraph of $L_G$ that is
  induced by $E_S$.
\end{definition}

\begin{definition}
We say that $v_S, w_S \in V_S$ are compatible, if (i) a child of
$v_S$ is $\theta'$ related to a child of $w_S$ or (ii) $v_S~\tau~w_S$.
\end{definition}

Theorem~\ref{theorem_subdiv} and Proposition~\ref{lemma:tau1} yield
the following characterization of convex cuts in terms of the search
graph $S$.

\begin{lemma}
\label{charaS} Let $C$ be a non-cyclic [cyclic] cut-set of a cut of
$G$ through $e_0$. Then the cut is convex if and only if $\gamma(C)$
is a maximal path [cycle] in $S$ such that any pair of vertices $v_S
\neq w_S$ on the path [cycle] is compatible.
\end{lemma}

If the cut-set $C$ of a cut of $G$ is non-cyclic, we have $C \cap
E(F_{\infty}) = \{e_0, e_{\vert C \vert-1}\}$, and if $C$ is cyclic,
we have $C \cap E(F_{\infty}) = \emptyset$.

The two matrices defined next will allow us to check in constant time
whether two vertices of $S$ are compatible. We build a $\vert E \vert
\times \vert E \vert$ matrix $A_{\tau}$ with boolean entries such that
$A_{\tau}(i,j)$ is true if and only if edge $i$ is $\tau$-related to
edge $j$. Likewise, we build a $(2 \vert E \vert) \times 2(\vert E
\vert)$ matrix $A_{\theta'}$ with boolean entries such that
$A_{\theta'}(i,j)$ is true if and only if edge $i$ in $G'$ is
$\theta'$-related to edge $j$ in $G'$.

Our algorithm for finding (the cut-sets of) all convex cuts of $G$
consists of two steps: find the non-cyclic cut-sets starting at each
$e_0 \in E(F_{\infty})$ and then find the cyclic ones starting at each
$e_0 \notin E(F_{\infty})$. In both steps we carry along and extend
paths $(e_0, \dots, e_k)$ of $S$ as long as all its vertices are pairwise
compatible. If $e_0 \in E(F_{\infty})$, there exists only one bounded
face $F_0$ whose boundary contains $e_0$ and the candidates for
$e_1$. If $e_0 \notin E(F_{\infty})$, there exist two such faces, and
we arbitrarily declare one of them to be $F_0$.

If, after starting at $e_0 \in E(F_{\infty})$, a path we carry along
has reached $E(F_{\infty})$ again, we have found a non-cyclic oriented
cut-set of a convex cut and store it (recall that the vertices of the
path are pairwise compatible). When we are done with $e_0$, \ie when
none of the paths that we carry along can be extended, we insert $e_0$
into a tabu list for further searches (we have already found all
convex cuts through $e_0$). The tabu list ensures that we do not end
up with two copies of a convex cut, \ie one for each orientation. Any
non-cyclic cut-set of a convex cut is found by our algorithm since it
carries along one orientation of any pairwise compatible non-cyclic
cut-set starting at $e_0$.

If, after starting at $e_0 \notin E(F_{\infty})$, a path we carry
along has reached $e_0$ again, we have found a cyclic cut-set of a
convex cut and store it. Conversely, any cyclic cut-set of a convex
cut is found by our algorithm since it carries along one orientation
of any pairwise compatible cyclic cut-set starting at $e_0$. Here,
the orientation is given by the choice of $F_0$ (see above). Again, we
insert $e_0$ into a tabu list.

\begin{algorithm}[H]
\caption{Finding the cut-sets $C$ of all convex cuts of a plane graph $G$}
\label{algo:plane}
\begin{algorithmic}[1]

\State Build the search graph and the matrices $A_{\tau}$ and
$A_{\theta'}$.

\State For any start vertex $e_0$ of $S$ with $e_0 \in E(F_{\infty})$
perform a breadth-first-traversal (BFT) starting at $e_0$. The first
path carried along is $(e_0)$. For any new vertex $v_S$ of $S$ that is
visited by the BFT and that is not in the tabu list, and for any path
and cycle carried along, use the matrices $A_{\tau}$ and $A_{\theta'}$
to check whether $v_s$ is compatible with the path or the
cycle. Whenever $E(F_{\infty})$ is reached, store $C$ and put $e_0$
into the tabu list.

\State For any $e_0 \notin E(F_{\infty})$ declare one of the two
bounded faces with $e_0$ on their boundaries to be $F_0$. Proceed as
in the case $e_0 \in E(F_{\infty})$, except that (i) when at $e_0$,
the BFT is restricted such that only edges in $E(F_0)$ are found and
(ii) $C$ is stored only if $e_0$ is reached. Finally, put $e_0$ into
the tabu list.
\end{algorithmic}
\end{algorithm}

\begin{theorem}
\label{thm:algo-time-plane}
Algorithm~\ref{algo:plane} finds all convex cuts of $G=(V,E)$ using
$\bigO(\vert V \vert^7)$ time and $\bigO(\vert V \vert^5)$ space.
\end{theorem}

\begin{proof}
Algorithm~\ref{algo:plane} is correct due to Lemma~\ref{charaS} and
the fact that we find any maximal path and cycle in $S$ with pairwise
compatible vertices exactly once.

To build the matrix $A_{\tau}$, we iterate over all vertices $v_S$ of
$S$ and identify all vertices of $S$ that are $\tau$-related to
$v_S$. For a given vertex $v_S$ this can be done in $\bigO(\vert V
\vert)$ time. Indeed, if $e = \{u, v\}$ is the edge in $G$ that equals
$v_S$, we can compute the distances of any $w \in V$ to $u$ and $v$ in
$\bigO(\vert V \vert)$, \eg by using BFT. For any $f \in E$ we can
then determine in constant time whether $e~\tau~f$.

To build the matrix $A_{\theta'}$, we proceed as above, except that we
use $G'$ instead of $G$. The running time for computing $A_{\tau}$ and
$A_{\theta'}$ is $\bigO(\vert E \vert^2)$, and $A_{\tau}$
and $A_{\theta'}$ take $\bigO(\vert E \vert^2)$ space.

Any time the BFT reaches a new vertex $v_S$ of $S$, the paths and
cycles carried along need to be checked for compatibility with
$v_s$. Using the matrices $A_{\tau}$ and $A_{\theta'}$, this takes
$\bigO(\vert E \vert)$ time per path. According to
Proposition~\ref{prop:upper_bound} the number of convex cuts of $G$
through a starting edge $e_0$ is bounded by $\vert E \vert^4$. This is
also the maximal number of paths that we carry along and that need to
be checked. Hence, processing $v_S$ takes time $\bigO(\vert E
\vert)^5$. Finding all convex cuts through a starting edge can then be
done in $\bigO(|E|^6)$ time and, since there are $\vert E \vert$
starting edges, total running time is $\bigO(\vert E \vert^7)$. Storing the
$\vert E \vert^4$ paths and cycles takes $\bigO(\vert E \vert^5)$
space.

The time and space requirements of (constructing) the search graphs
and the tabu lists are below the time and space requirements specified
so far. The claim now follows from $\bigO(\vert E \vert) = \bigO(\vert
V \vert)$ which, in turn, is a consequence of $G$ being plane.
\end{proof}


%
%
%

\section{Conclusions}
\label{sec:conclusions}
We have presented an algorithm for finding all convex cuts of a plane
graph in polynomial time. To the best of our knowledge, it is the
first polynomial-time algorithm for this task. We have also presented
an algorithm that computes all convex cuts of a not necessarily plane
but bipartite graph in cubic time.

Both algorithms are based on binary, symmetric, but generally not
transitive relations on edges. In the case of a plane graph $G$ we
employed two relations: (i) the \djoko on the edges of a subdivision
of $G$ and (ii) another relation on the edges of $G$. In case of a
bipartite graph it was sufficient to employ the \djoko on the graph's
edges.

To prove that the number of convex cuts of a plane graph is not
exponential, we employed results on the intersection pattern of convex
cuts that are based on a specific embedding of the cuts. Thus, a
connection to the first part of the paper arises, where we defined a
sub-class of plane graphs via the intersection patterns of certain
embedded cuts (which all turned out to be convex). In particular, the
transition from the sub-class to general plane graphs is reflected by
a generalization of the intersection patterns of convex cuts from
arrangements of pseudolines to patterns where forks and joins of
convex cuts are possible.

The characterization of convex cuts of general graphs, as given by
Theorem~\ref{theorem_subdiv}, was instrumental in finding all convex
cuts of a bipartite or a plane graph in polynomial time. We reckon
that this new characterization of convex cuts of general graphs also
helps when devising polynomial-time algorithms for finding convex cuts
in graphs from other classes.

\paragraph*{Acknowledgments.}
We thank our colleagues Andreas Gemsa, Peter Sanders, and Christian Schulz for helpful
discussions on the topic.

\bibliographystyle{abbrv}
\bibliography{roland,paper,paper2}

\begin{thebibliography}{10}

\bibitem{Artigas20111968}
D.~Artigas, S.~Dantas, M.~Dourado, and J.~Szwarcfiter.
\newblock Partitioning a graph into convex sets.
\newblock {\em Discrete Mathematics}, 311(17):1968 -- 1977, 2011.

\bibitem{Balister2009509}
P.~Balister, S.~Gerke, G.~Gutin, A.~Johnstone, J.~Reddington, E.~Scott,
  A.~Soleimanfallah, and A.~Yeo.
\newblock Algorithms for generating convex sets in acyclic digraphs.
\newblock {\em Journal of Discrete Algorithms}, 7(4):509 -- 518, 2009.

\bibitem{bichot2011graph}
C.~Bichot and P.~Siarry.
\newblock {\em Graph Partitioning}.
\newblock Wiley, 2011.

\bibitem{Bjoerner99a}
A.~{Bj\"{o}rner}, M.~{Las Vergnas }, B.~{Sturmfels}, N.~{White}, and
  G.~{Ziegler}.
\newblock {\em Oriented Matroids}.
\newblock Encyclopedia of Mathematics and its Applications, 46. Cambridge
  University Press, second edition, 1999.

\bibitem{BulucMSSS13recent}
A.~Bulu\c{c}, H.~Meyerhenke, I.~Safro, P.~Sanders, and C.~Schulz.
\newblock Recent advances in graph partitioning.
\newblock eprint arXiv:1311.3144, \url{http://arxiv.org/abs/1311.3144}, Dec
  2013.

\bibitem{Chepoi97a}
V.~{Chepoi}, M.~{Deza}, and V.~{Grishukhin}.
\newblock {Clin d'{\oe}il on $L_1$-embeddable planar graphs}.
\newblock {\em Discrete Applied Mathematics}, 80:3--19, 1997.

\bibitem{DBLP:conf/ipps/DellingGRW11}
D.~Delling, A.~V. Goldberg, I.~Razenshteyn, and R.~F.~F. Werneck.
\newblock Graph partitioning with natural cuts.
\newblock In {\em Proc. 25th IEEE Intl. Parallel and Distributed Processing
  Symposium (IPDPS'11)}, pages 1135--1146, 2011.

\bibitem{Diestel2006a}
R.~Diestel.
\newblock {\em Graph Theory}.
\newblock Graduate texts in mathematics. Springer, 2006.

\bibitem{Djokovic73a}
D.~{Djokovi\'{c}}.
\newblock {Distance-Preserving Subgraphs of Hypercubes}.
\newblock {\em Journal of Combinatorial Theory B}, 14:263--267, 1973.

\bibitem{DouradoPRS12convexity}
M.~Dourado, F.~Protti, D.~Rautenbach, and J.~Szwarcfiter.
\newblock On the convexity number of graphs.
\newblock {\em Graphs and Combinatorics}, 28:333--345, 2012.

\bibitem{Eppstein2006a}
D.~{Eppstein}.
\newblock {Cubic Partial Cubes from Simplicial Arrangements}.
\newblock {\em The electronic journal of combinatorics}, 13:$\#$R79, 2006.

\bibitem{Glantz2013c}
R.~{Glantz} and H.~{Meyerhenke}.
\newblock Finding all convex cuts of a plane graph in cubic time.
\newblock In {\em Proc. 8th International Conference on Algorithms and
  Complexity (CIAC'13)}, pages 246--263, 2013.

\bibitem{MeyerhenkeMS09new}
H.~Meyerhenke, B.~Monien, and T.~Sauerwald.
\newblock A new diffusion-based multilevel algorithm for computing graph
  partitions.
\newblock {\em Journal of Parallel and Distributed Computing}, 69(9):750--761,
  2009.
\newblock Best Paper Awards and Panel Summary: IPDPS 2008.

\bibitem{Ovchinnikov2008a}
S.~{Ovchinnikov}.
\newblock {Partial cubes: Structures, characterizations, and constructions}.
\newblock {\em Discrete Mathematics}, 308:5597--5621, 2008.

\bibitem{Peterin2008a}
I.~Peterin.
\newblock A characterization of planar partial cubes.
\newblock {\em Discrete Mathematics}, 308(24):6596--6600, 2008.

\bibitem{Thomassen92a}
C.~{Thomassen}.
\newblock {The Jordan-Sch\"{o}nflies Theorem and the Classification of
  Surfaces}.
\newblock {\em The American Mathematical Monthly}, 99:116--130, 1992.

\bibitem{Wilkeit90a}
E.~{Wilkeit}.
\newblock {Isometric Embeddings in Hamming Graphs}.
\newblock {\em Journal of Combinatorial Theory B}, 14:179--197, 1990.

\end{thebibliography}

\end{document}